\documentclass[]{new_tlp}

\usepackage{upgreek,xcolor,bm,tikz,stackrel,relsize}
\usepackage{units,stmaryrd}
\usepackage{amsmath,amssymb}
\usepackage[nointegrals]{wasysym}
\usepackage[inline]{enumitem}
\usepackage{float}
\usepackage{subcaption}
\usepackage{stackengine}
\usepackage{chngcntr}
\counterwithin{figure}{section}
\counterwithin{table}{section}
\usetikzlibrary{arrows,automata}

\usepackage{multicol}

\newcommand{\longversion}[1]{}

\newcommand{\timesn}[1]{#1_{ \infty} } 
\newcommand{\pimor}{f}

\newcommand{\logbasicb}{{\sf ITL} _{\ubox}}
\newcommand{\logexpb}{{\sf CDTL}_{\ubox}}
\newcommand{\loghomeob}{{\sf ITL}^{+}_{\ubox}}
\newcommand{\logpersb}{{\sf CDTL}^{+}_{\ubox}}

\newcommand{\logrealaxa}{{\sf RTL}_{\diam{\khence}} }
\newcommand{\logrnaxa}{{\sf ETL}_{\diam{\khence}} }
\newcommand{\logrnfsa}{{\sf ETL}^+_{\diam{\khence}}}

\newcommand{\logbasica}{{\sf ITL} _{\diam{\khence}}}
\newcommand{\logexpa}{{\sf CDTL}_{\diam{\khence}}}
\newcommand{\loghomeoa}{{\sf ITL}^{+}_{\diam{\khence}}}
\newcommand{\logpersa}{{\sf CDTL}^{+}_{\diam{\khence}}}

\newcommand{\ignore}[1]{}

\newcommand{\AxCons}[2]{{\rm CD}(#1,#2)}
\newcommand{\AxBInd}[2]{{\rm BI}(#1,#2)}
\newcommand{\AxConsm}[1]{{\rm CD}^-(#1 )}

\newcommand{\logrealax}{{\sf RTL}_{\diam\ubox} }
\newcommand{\logrnax}{{\sf ETL}_{\diam\ubox} }
\newcommand{\logrnfs}{{\sf ETL}^+_{\diam\ubox}}
\newcommand{\wlogbasic}{{\sf ITL}^0_{\diam\ubox}}
\newcommand{\wboxfix}{\rm WH}
\newcommand{\logbasic}{{\sf ITL} _{\diam\ubox}}

\newcommand{\loghomeo}{{\sf ITL}^+_{\diam\ubox}}
\newcommand{\logexp}{{\sf CDTL}_{\diam\ubox} }
\newcommand{\logpers}{{\sf CDTL}^+_{\diam\ubox}}
\newcommand{\khence}{\boxast}

\newcommand{\cl }{\mathcal }

\newcommand{\itlc}{{\sf ITL}^{\sf c}_{\diam\ubox}}
\newcommand{\itle}{{\sf ITL}^{\sf e}_{\diam\ubox}}
\newcommand{\itlp}{{\sf ITL}^{\sf p}_{\diam\ubox}}
\newcommand{\itlo}{{\sf ITL}^{\sf o}_{\diam\ubox}}

\newcommand{\itleb}{{\sf ITL}^{\sf e}_{ \ubox}}
\newcommand{\itlpb}{{\sf ITL}^{\sf p}_{ \ubox}}

\newcommand{\itlca}{{\sf ITL}^{\sf c}_{\diam{{\khence}}}}

\newcommand{\tnext}{{\ocircle}}
\newcommand{\mtnext}{\circ}
\newcommand{\ubox}{\Box}

\newcommand{\diam}{\Diamond}
\newcommand{\val}[1]{\lb #1 \rb}

\newcommand{\peq}{\preccurlyeq}
\newcommand{\seq}{\succcurlyeq}

\newcommand{\acc}{\peq}

\def\lb{\left\llbracket}
\def\rb{\right\rrbracket}
\def\<{\left (}

\def\>{\right )}
\def\({\left (}
\def\){\right )}

\def\cbra{\left \{}
\def\cket{\right \}}

\def\eqdef{\stackrel{\rm def}{=}}

\usepackage{hyperref}
\usepackage{graphics}
\usepackage{scalerel}

\usetikzlibrary{svg.path}

\title{Exploring the Jungle of Intuitionistic Temporal Logics}
\author[Boudou et al.]{JOSEPH BOUDOU\\
					   IRIT, Toulouse University, Toulouse, France.\\ \email{joseph.boudou@matabio.fr}\\
					   \and MART\'IN DI\'EGUEZ\\
					   LERIA. University of Angers. Angers, France.\\ \email{martin.dieguezlodeiro@univ-angers.fr} \\
					   \and DAVID FERN\'ANDEZ-DUQUE\\
					   Department of Mathematics, Ghent University. Ghent, Belgium.\\\email{David.FernandezDuque@UGent.be}\\
				\and	   PHILIP KREMER\textsuperscript{\thanks{This
research was supported by the Social Sciences and Humanities Research
Council of Canada.}}\\
					   Department of Philosophy, University of Toronto, Toronto, Canada.\\ \email{philip.kremer@utoronto.ca}}

\jdate{March 2003}
\pubyear{2003}
\pagerange{\pageref{firstpage}--\pageref{lastpage}}
\doi{S1471068401001193}

\newtheorem{theorem}{Theorem}[section]
\newtheorem{definition}[theorem]{Definition}
\newtheorem{example}[theorem]{Example}
\newtheorem{lemma}[theorem]{Lemma}

\newtheorem{corollary}[theorem]{Corollary}
\newtheorem{proposition}[theorem]{Proposition}
\newtheorem{remark}[theorem]{Remark}
\newtheorem{question}{Question}

\begin{document}

\maketitle


\begin{abstract}
The importance of intuitionistic temporal logics in Computer Science and Artificial Intelligence has become increasingly clear in the last few years. 
From the proof-theory point of view, intuitionistic temporal logics have made it possible to extend functional programming languages with new features via type theory, while from the semantics perspective
several logics for reasoning about dynamical systems and several semantics for logic programming have their roots in this framework.
We consider several axiomatic systems for intuitionistic linear temporal logic and show that each of these systems is sound for a class of structures based either on Kripke frames or on dynamic topological systems.
We provide two distinct interpretations of `henceforth', both of which are natural intuitionistic variants of the classical one.
We completely establish the order relation between the semantically-defined logics based on both interpretations of `henceforth', and, using our soundness results, show that the axiomatically-defined logics enjoy the same order relations.
Under consideration in Theory and Practice of Logic Programming (TPLP).
\end{abstract}

\section{Introduction}

Intuitionistic logic ($\sf IL$) e.g.~\cite{Heyting1930,MintsInt} enjoys a myriad of interpretations based on computation, information or topology, making it a natural framework to reason about dynamic processes in which these phenomena play a crucial role.
Thus, it should not be surprising that combinations of intuitionistic logic and linear temporal logic ($\sf LTL$)~\cite{Pnu77} have been proposed for applications within several different contexts:

\paragraph{Types for functional programming languages.} The Curry-Howard correspondence identifies intuitionistic proofs with the $\lambda$-terms of functional programming~\cite{Howard80}. Several extensions of the $\lambda$-calculus with operators from $\sf LTL$ have been proposed in order to introduce new features to functional programming languages:
\cite{Davies96,Davies-2017} has suggested adding a `next' ($\tnext$) operator to $\sf IL$ in order to define the type system $\lambda^\tnext$, which allows extending functional programming languages with  \emph{staged computation}\footnote{\emph{Staged computation} is a technique that allows dividing the computation in order to exploit the early availability of some arguments.}~\cite{Ershov77}.
\cite{Davies01} proposed the functional programming language ${\sf Mini\text{-}ML}^\Box$ which is supported by intuitionistic $\sf S4$ and allows capturing complex forms of staged computation as well as runtime code generation.
\cite{Yuse2006} later extended $\lambda^\tnext$ to $\lambda^\Box$ by incorporating the `henceforth' operator ($\ubox$), useful for modelling persistent code that can be executed at any subsequent state.  

\paragraph{Semantics for dynamical processes.} Intuitionistic temporal logics have been proposed as a tool for modelling semantically-given processes.
\cite{Maier2004Heyting} observed that an intuitionistic temporal logic with `henceforth' and `eventually'\footnote{In this paper, `eventually' should be understood as `occurring at least once, either now or in the future,' while `henceforth' should be understood as `from now on.'}  ($\diam$) could be used for reasoning about safety and liveness conditions in possibly-terminating reactive systems, and
\cite{FernandezITLc} has suggested that a logic with `eventually' can be used to provide a decidable framework in which to reason about topological dynamics.

\paragraph{Temporal answer set programming.} In the areas of {nonmonotonic reasoning,} {knowledge representation (KR),} and {artificial intelligence,} intuitionistic and intermediate logics have played an important role within the successful {answer set programming (ASP)}~\cite{Brewka11} paradigm for KR, leading to several extensions of modal ASP~\cite{CP07} that are supported by intuitionistic-based modal logics like \emph{temporal here and there}~\cite{BalbianiDieguezJelia}.
\bigskip

Despite interest in the above applications, there is a large gap to be filled regarding our understanding of the computational behaviour of intuitionistic temporal logics.
We have successfuly employed semantical methods to show the decidability of the logic $\sf ITL^e$ defined by a natural class of Kripke frames \cite{BoudouCSL} and shown that these semantics correspond to a natural calculus over the $\ubox$-free fragment \cite{DieguezCompleteness}.
However, as we will see, in the presence of $\ubox$, new validities arise which may be undesirable from the point of view of an extended Curry-Howard isomorphism.
Thus, our goal is to provide semantics for weaker axiomatically-defined intuitionistic temporal logics in order to provide tools for understanding their computational behaviour.
We demonstrate the power of our semantics by separating several natural axiomatically-given calculi, which in particular answers in the negative a conjecture of \cite{Yuse2006} that the Gentzen-style and the Hilbert-style calculi presented there prove the same set of formulas.

There have already been some notable efforts towards a semantical study of intuitionistic temporal logics.
\cite{KojimaNext} endowed Davies's logic with Kripke semantics and provided a complete deductive system.
Bounded-time versions of logics with henceforth were later studied by \cite{KamideBounded}.
Both use semantics based on Simpson's bi-relational models for intuitionistic modal logic \cite{Simpson94}.
Since then, \cite{BalbianiDieguezJelia} have shown that temporal here-and-there is decidable and enjoys a natural axiomatization.
Topological semantics for intuitionistic modal and tense logics have also been studied by \cite{Davoren2009,DavorenIntuitionistic}, and \cite{KremerIntuitionistic} suggested a topologically-defined intuitionistic variant of $\sf LTL$ with $\tnext$ and an intuitionistic variant of `henceforth', which we will denote by ${\khence}$.
Kremer's proposal was never published, and is presented from the axiomatic perspective in Section \ref{SecBasic} and its semantics in Section~\ref{SecTopre}. The decidability of the logic of the weak semantics remains open, but \cite{FernandezITLc} has shown that a similar logic with `eventually' $\diam$ instead of $\ubox$ is decidable.
%
%
\medskip

In this paper we lay the groundwork for an axiomatic treatment of intuitionistic linear temporal logics.
We will introduce a `basic' intuitionistic temporal logic, $\logbasic$, defined by adding standard axioms of $\sf LTL$ to intuitionistic modal logic (see Section \ref{SecBasic} for details).
We also consider additional Fischer Servi axioms ({\rm FS}), a `constant domain' axiom ${\rm CD}:=\ubox (p\vee q) \to \ubox p \vee \diam q$, and a `conditional excluded middle' axiom ${\rm CEM}:=( \neg \tnext p \wedge \tnext \neg \neg p ) \rightarrow ( \tnext q \vee \neg \tnext q )$.
Combining these, we obtain seven intuitionistic temporal logics.
Logics with the constant domain axiom are sound for their Kripke semantics, given by the class of dynamical systems based on a poset, also called {\em expanding posets.}
In the setting of Kripke semantics, the Fischer Servi axioms correspond to backwards-confluence of the transition function.

The constant domain axiom is not derivable from the others, and to show this, we will consider topological semantics for intuitionistic temporal logic.
This is in contrast to the setup in \cite{BalbianiToCL}, where all logics are based on expanding posets and hence validate $\rm CD$.
The crucial difference between working with expanding posets vs.~topological semantics is that the standard interpretation of `henceforth' in classical $\sf LTL$ may readily be applied to the setting of expanding posets, whereas in the topological setting, it requires some modification.

Recall that topological spaces are pairs $(X,\mathcal T)$, where $\mathcal T$ is a collection of subsets of $X$ closed under unions and finite intersections (see Section \ref{SecSemantics}).
Elements of $\mathcal T$ are called {\em open sets.}
In the topological semantics of intuitionistic logic, each proposition $\varphi$ must be interpreted as an open set $\val\varphi\subseteq X$.
In order to interpret tenses, we equip $(X,\mathcal T)$ with a continuous function $S\colon X\to X$.
The classical semantics for next and eventually yield well-defined operations in this setting: for example, we define $\val{\tnext\varphi} = S^{-1}\val\varphi$, which amounts to the standard definition where $x\in \val{\tnext\varphi}$ iff $S(x) \in \val \varphi$.
The continuity of $S$ ensures that $\val{\tnext\varphi}$ is an open set whenever $\val{ \varphi}$ is (recall that by definition, $S$ is continuous iff preimages of open sets are open). 
Similarly, setting $x\in \val{\diam\varphi}$ iff there is $n\geq 0$ such that $S^n(x)\in\val\varphi$ ensures that $\val{\diam\varphi}$ will always be open.

However, the classical definition of $\val{\ubox\varphi}$ would have that $x\in\val{\ubox\varphi}$ iff $S^{n}(x) \in \val\varphi$ for all $n\geq 0$ or, equivalently, $\val{\ubox\varphi} = \bigcap_{n\geq 0} \val\varphi$.
The problem is that open sets need not be closed under infinite intersections, so an intuitionistic interpretation for $\ubox\varphi$ must modify the classical semantics in a way that only open sets are produced.
There are at least two ways to achieve this.
We call these the `weak' and `strong' interpretations of $\ubox$.
The first was originally proposed by \cite{KremerIntuitionistic} in an unpublished note, and is treated similarly to the universal quantifier in the context of intuitionistic semantics of first order logic.
In order to distinguish it from the strong interpretation, we will denote it by $\khence$.\footnote{\cite{KremerIntuitionistic} instead uses $\ast$.}
As we will see, the operator ${\khence}$ does not satisfy some key $\sf LTL$ validities, namely ${\khence} p\to\tnext {\khence} p$, ${\khence}\tnext p\to \tnext {\khence} p$, and ${\khence} p\to {\khence}{\khence} p$. Consequently, some of the standard $\sf LTL$ axioms are not sound for this interpretation.
We thus propose a logic $\logbasica$, where the axiom ${\khence} p\to\tnext {\khence} p$ is replaced by the weaker ${\khence} p\to {\khence}  \tnext p$.

Nevertheless, $\ubox p \to \tnext\ubox p$ is arguably one of the defining axioms for henceforth, so it is convenient to have semantics that validate it.
In order to obtain semantics for $\logbasic$, we propose a new interpretation for $\ubox$.
Our approach is natural from an algebraic perspective, as we define the interpretation of $\ubox \varphi$ via a greatest fixed point in the Heyting algebra of open sets.\footnote{We will not discuss Heyting algebras in this text, but see e.g.~\cite{Heyting1930,MintsInt}.}
We will show that dynamic topological systems provide semantics for the logics without the constant domain axiom, from which we conclude the independence of the latter. Moreover, we show that the Fischer Servi axioms are valid for the class of {\em open} dynamical topological systems, and that in this setting, the semantics for ${\khence}$ and $\ubox$ coincide.
The constant domain axiom shows that the $\{\diam,\ubox\}$-logic of expanding posets is different from that of dynamic topological systems. We show via an alternative axiom that the $\{\tnext,\ubox\}$-logics are also different.
We also consider the special case where topological semantics are based on Euclidean spaces.
We show that this leads to logics strictly between that of all spaces and that of expanding posets.
In the special case of the real line, we can prove that every formula falsified on a persistent poset is falsifiable on the real line.

\medskip

\noindent{\sc Layout.} Section \ref{SecBasic} introduces the syntax and the axiomatic systems as well as its weak counterparts that we propose for intuitionistic temporal logic. 
Section \ref{SecTopre} reviews dynamic topological systems, which are used in Section \ref{SecSemantics} to provide semantics for our formal language. 
Section \ref{SecSound} shows that four of our logics and their weak companions are each sound for a class of dynamical systems, and Section \ref{secEuclid} shows that the remaining logics are sound for Euclidean spaces. In Section~\ref{secPers} we focus on $\logbasic$ interpreted on persistent posets and its connection with the real line. In Section \ref{SecInd} we show that several of the logics we consider are pairwise distinct. Finally, Section \ref{SecConc} lists some open questions.

\section{Syntax and axiomatics}\label{SecBasic}

In this section we will introduce several natural intuitionistic temporal logics. Most of the axioms we consider have appeared either in the intuitionistic logic, the temporal logic, or the intuitionistic modal logic literature. They will be based on the language of linear temporal logic, as defined next.

Fix a countably infinite set $\mathbb P$ of {\em propositional variables.} The full language $\mathcal L_{\diam\ubox{\khence}}$ of intuitionistic (linear) temporal logic $\sf ITL$ is given by the grammar in Backus-Naur form
\[ \varphi,\psi := \ \  \bot \  | \   p  \ |  \ \varphi\wedge\psi \  |  \ \varphi\vee\psi  \ |  \ \varphi\to\psi  \ |  \ \tnext\varphi \  | \  \diam\varphi \  |  \ \ubox \varphi  \  |  \ {\khence}\varphi, \]
where $p\in \mathbb P$. As usual, we use $\neg\varphi$ as a shorthand for $\varphi\to \bot$ and $\varphi \leftrightarrow \psi$ as a shorthand for $(\varphi \to \psi) \wedge (\psi \to \varphi)$. We read $\tnext$ as `next', $\diam$ as `eventually', $\ubox$ as `strong henceforth' and ${\khence}$ as `weak henceforth'.
The intuition is that formulas are evaluated at moments of discrete time. The formula $\tnext \varphi$ indicates that $\varphi$ will hold at the next moment, $\diam \varphi$ that it will hold in some subsequent moment, and $\ubox\varphi$ and ${\khence}\varphi$ both indicate that $\varphi$ will hold in every subsequent moment, including the current moment.
However, as we will see, making sense of the latter notion in intuitionistic semantics is not straightforward, thus giving rise to two natural, but distinct, interpretations.

Given any formula $\varphi$,
we denote the set of subformulas of $\varphi$ by ${\mathrm{sub}}(\varphi)$.
For $\Theta \subseteq \{\diam,\ubox,{\khence}\}$, the language $\mathcal L_\Theta$ is the sub-language of $\mathcal L_{\diam\ubox{\khence}}$ whose only tenses are $\tnext$ and those in $\Theta$; we will not consider languages without $\tnext$.
So, for example, $\mathcal L_\diam$ only has tenses $\tnext $ and $\diam$.
We will write $\mathcal L_\mtnext$ instead of $\mathcal L_\varnothing$.

The tenses ${\khence}$, $\ubox$ represent two possible intuitionistic readings of `henceforth' and thus we will rarely consider logics with both.
In order to compare logics based on ${\khence}$ with logics based on $\ubox$, we introduce the translations $t_{\khence}$, where $t_{\khence}(\varphi) \in \mathcal L_{\diam {\khence}}$ is the formula obtained by replacing every occurrence of $\ubox$ in $\varphi$ by ${\khence}$, and similarly define $t_\ubox$, which replaces every occurrence of ${\khence}$ by $\ubox$.
The semantics for ${\khence}$ first appeared in the unpublished note \cite{KremerIntuitionistic}, while those for $\ubox$ were first introduced in a preliminary version of this paper \cite{BoudouJelia}.


We begin by establishing our basic axiomatization for logics over $\mathcal L_{\diam\ubox}$.
It is obtained by adapting the standard axioms and inference rules of $\sf LTL$ \cite{temporal}, as well as their dual versions.

\begin{definition}\label{defLogbasic}
The logic $\logbasic$ is the least set of $\mathcal L_{\diam\ubox}$-formulas closed under the following axioms and rules.

\begin{multicols}{2}
\begin{enumerate}[label=({\sc \roman*})]
\item\label{ax01Taut} All intuitionistic tautologies;
\item\label{ax02Bot} $\neg \tnext \bot$;
\item\label{ax03NexWedge} $\tnext \left( \varphi \wedge \psi \right) \leftrightarrow \left(\tnext \varphi \wedge\tnext \psi\right)$;
\item\label{ax04NexVee} $\tnext \left( \varphi \vee \psi \right) \leftrightarrow \left(\tnext \varphi \vee\tnext \psi\right)$;
\item\label{ax05KNext} $\tnext\left( \varphi \rightarrow \psi \right) \rightarrow \left(\tnext\varphi \rightarrow \tnext\psi\right)$;
\item\label{ax06KBox} $\ubox \left( \varphi \rightarrow \psi \right) \rightarrow \left(\ubox \varphi \rightarrow \ubox \psi\right)$;
\item\label{ax07:K:Dual} $\ubox \left( \varphi \rightarrow \psi \right) \rightarrow \left(\diam \varphi \rightarrow \diam \psi\right)$;
\item\label{ax09BoxT} $\ubox \varphi \to \varphi  $;
\item\label{ax09BoxFix} $\ubox \varphi \to   \tnext \ubox \varphi$;
\item\label{ax10DiamT} $\varphi   \to \diam \varphi$;

\item\label{ax10DiamFix} $ \tnext \diam \varphi \to \diam \varphi$;
\item\label{ax11:ind:1} $\ubox ({ \varphi \rightarrow \tnext \varphi } )\to ({ \varphi \rightarrow \ubox \varphi })$;
\item\label{ax12:ind:2} $\ubox ({ \tnext \varphi \to \varphi})\to ({ \diam \varphi \rightarrow \varphi } )$;
\item\label{ax13MP} $\dfrac{\varphi \ \ \varphi\to \psi}\psi$;
\item\label{ax14NecCirc} $\dfrac\varphi {\tnext\varphi}$, $\dfrac\varphi {\ubox\varphi}$.
\end{enumerate}
\end{multicols}
\end{definition}
Axioms \ref{ax05KNext} and \ref{ax06KBox} hold in any normal modal logic and \ref{ax07:K:Dual} is a dual version of \ref{ax06KBox}; such dual axioms are often needed in intuitionistic modal logic, since $\diam$ and $\ubox$ are not typically inter-definable.
The axioms \ref{ax02Bot}-\ref{ax04NexVee} have to do with the passage of time being deterministic in linear temporal logic, and are related to a functional modality, i.e.~a modality that is interpreted using a function rather than a relation.
The axioms \ref{ax09BoxT} and \ref{ax10DiamT} have to do with future tenses being interpreted reflexively, i.e.~$\varphi$ is considered to hold eventually if it holds now.
The axiom \ref{ax10DiamFix} states that if something will henceforth be the case, then in the next moment, it will still henceforth be the case, and \ref{ax11:ind:1} is successor induction, as time is interpreted over the natural numbers.
Axioms \ref{ax10DiamFix} and \ref{ax12:ind:2} are their duals.
All rules are standard in any normal modal logic.

Each axiom is either included in the axiomatization of Goldblatt \cite[page 87]{Goldblatt92} or is a variant of one of them (e.g., a contrapositive); this is standard in intuitionistic modal logic, as such variants are needed to account for the independence of the basic connectives.
We do not consider `until' and `release' in this paper, but these operators have previously studied within an intuitionistic context in~\cite{BalbianiToCL}.

Next we define our base logic for weak henceforth.
It is convenient to present it as a logic over $\mathcal L_{\diam\ubox}$ and then translate to $\mathcal L_{\diam\khence}$.
The main reason for this is that we view the weak and strong semantics of henceforth as two possible interpretations of intuitionistic temporal logic, rather than two independent tenses.
From this point of view, the notation $\khence$ should be seen as an indication that weak semantics are being used.
Moreover, we are interested in comparing logics based on $\khence$ with those based on $\ubox$, and uniform notation will be helpful for this.
In particular, we will see that the weak semantics give rise to weaker logics, partially motivating the terminology.
We will then use the translation $t_{\khence}$ (which, recall, replaces $\ubox$ by $\khence$) when we wish to indicate that we are working with weak semantics.

We define $\wlogbasic$ as $\logbasic$, but replacing axiom \ref{ax09BoxFix} by
\[\wboxfix \eqdef  \ubox\varphi \to \ubox \tnext \varphi.\]
Here, $\wboxfix$ stands for `weak henceforth'.
This terminology is justified by the following.

\begin{lemma}\label{lemmaWHF}
Every instance of $\wboxfix$ is derivable in $\logbasic $.
\end{lemma}

\begin{proof}
From~\ref{ax09BoxT},~\ref{ax14NecCirc} and~\ref{ax06KBox} we obtain $\tnext \ubox \varphi \to \tnext \varphi$, which combined with axiom~\ref{ax09BoxFix}	allows us to conclude $\vdash \ubox \varphi \to \tnext \varphi$.
Thanks to~\ref{ax14NecCirc} and axiom~\ref{ax06KBox} we get 
\begin{equation}
\vdash \ubox \ubox \varphi \to \ubox \tnext \varphi. \label{eq:1}
\end{equation}
From axiom~\ref{ax14NecCirc}, necessitation and~\ref{ax11:ind:1} we obtain 
\begin{equation}
\vdash \ubox \left( \ubox \varphi \to \tnext \ubox \varphi \right) \to \left( \ubox \varphi \to \ubox \ubox \varphi\right). \label{eq:2}
\end{equation}
From~\eqref{eq:1} and~\eqref{eq:2} we conclude $\vdash \ubox \varphi \to \ubox \tnext \varphi$.	
\end{proof}

Thus $\wlogbasic\subseteq \logbasic$.
We will later see that the inclusion is strict, since ${\khence} \varphi \to \tnext {\khence} \varphi$ is not valid in general for our semantics, but ${\khence} \varphi \to {\khence} \tnext  \varphi$ and $\ubox \varphi \to \tnext \ubox \varphi$ are.
We can then define $\logbasica = t_{\khence} (\wlogbasic) $.

Modal intuitionistic logics often involve additional axioms, and in particular \cite{FS84} includes the schema
\[{\rm FS}_\diam (\varphi,\psi) \eqdef \left( \diam \varphi \rightarrow \ubox \psi \right) \to \ubox \left(\varphi \rightarrow \psi\right)
.\]
We also define
\[{\rm FS}_\mtnext (\varphi,\psi)  \eqdef \left(\tnext \varphi \rightarrow \tnext \psi \right) \to \tnext \left(\varphi \rightarrow \psi\right).\]
This notation is justified in view of the fact that $\tnext$ is self-dual in linear temporal logic: since time is modeled deterministically, `in some next moment' is equivalent to `in every next moment,' so $\tnext$ may be regarded as a `box' or a `diamond.'
It is further motivated by the following.

\begin{proposition}\label{propFStoFS}	
The formula $ {\rm FS}_\mtnext (\diam p,\ubox q) \to {\rm FS}_\diam (p,q) $ is derivable in $\logbasic $.
\end{proposition}

\proof We reason within $\logbasic$. Assume \begin{enumerate*}[label=\arabic*)]
	\item\label{assumption:1} ${\rm FS}_\mtnext (\diam p,\ubox q)$ and 
	\item\label{assumption:2} $\left(\diam p \to \ubox q\right)$.
\end{enumerate*}
Notice that $\tnext \diam p \to \diam p$ and $\ubox p \to \tnext \ubox p$ are instances of axioms~\ref{ax09BoxFix} and~\ref{ax10DiamFix}. From this and assumption~\ref{assumption:2} we conclude $\tnext \diam p \to \tnext \ubox p$. Thanks to assumption~\ref{assumption:1} and Modus Ponens we conclude $\tnext \left(\diam p \rightarrow \ubox q\right)$. Therefore, $\left(\diam p \to  \ubox p\right) \to \tnext \left(\diam p \rightarrow \ubox q\right)$. By Rule~\ref{ax14NecCirc}, we obtain $ \ubox \left(\left(\diam p \to \ubox p\right) \to \tnext \left(\diam p \rightarrow \ubox q\right)\right)$. By the induction axiom~\ref{ax11:ind:1} we derive $\left(\diam p \to \ubox p\right) \to \ubox \left(\diam p \rightarrow \ubox q\right)$. From the assumption~\ref{assumption:2} and Modus Ponens it follows that $\ubox \left(\diam p \rightarrow \ubox q\right)$.
Note that $\left(\diam p \to \ubox q\right) \to \left(p \to q\right)$ is derivable in $\logbasic$. From rule~\ref{ax14NecCirc} and axiom~\ref{ax06KBox} we obtain $\ubox\left(\diam p \to \ubox q\right) \to \ubox \left(p \to q\right)$.
From this and $\ubox \left(\diam p \rightarrow \ubox q\right)$ it follows that $\ubox \left(p \to q\right)$, as required.
\endproof

Later we will show that these schemas lead to logics strictly stronger than $\logbasic$.
Next we consider additional axioms reminiscent of the constant domain axiom in first-order intuitionistic logic, namely $\forall x (\varphi (x) \vee \psi (x)) \rightarrow \exists x \varphi (x) \vee \forall x \psi(x)$; we maintain the terminology `constant domain' due to this similarity, although they do not retain this meaning in our logics.
As we will see, in the context of intuitionistic temporal logics, these axioms separate Kripke semantics from the more general topological semantics.
\begin{align*}
\AxCons \varphi\psi  &\eqdef \ubox( \varphi \vee \psi) \to \ubox \varphi \vee \diam \psi\\
\AxBInd \varphi\psi  &\eqdef \ubox ( \varphi \vee \psi) \wedge \ubox (\tnext \psi\rightarrow \psi) \rightarrow \ubox \varphi \vee \psi .
\end{align*}
Here, $\rm CD$ stands for `constant domain' and $\rm BI$ for `backward induction'.
We also define as a special case $\AxConsm \varphi = \AxCons{\neg\varphi}\varphi$.

The axiom $\rm CD$ might not be desirable from a constructive perspective, as from $\ubox ( \varphi \vee \psi)$ one cannot in general extract an upper bound for a witness for $\diam \psi$.\footnote{For example, if $\varphi$ represents the `active' states and $\psi$ the `halting' states of a program, then $\rm CD$ would require us to decide whether the program halts, which is not possible to do constructively.}
The axiom $\rm BI$ is a $\diam$-free version of $\rm CD$, as witnessed by the following.
\begin{proposition}\label{PropConstoBI}
The following formulas are derivable in $\wlogbasic$:
\begin{enumerate}

\item $\AxCons pq \rightarrow \AxBInd pq$;

\item $\AxBInd p{\diam q} \rightarrow \AxCons pq$.

\end{enumerate}
\end{proposition}

\begin{proof}
We reason within $\wlogbasic$.
For the first claim, assume that
\begin{enumerate*}[label=\arabic*)]
\item\label{ItDerivOne} $\AxCons pq$,
\item\label{ItDerivTwo} $\ubox (\tnext q\rightarrow q)$, and
\item\label{ItDerivThree} $\ubox ( p\vee q)$.
\end{enumerate*}
From \ref{ItDerivOne} and \ref{ItDerivThree} we obtain $\ubox p \vee \diam q$, which together with \ref{ItDerivTwo} and axiom \ref{ax12:ind:2} gives us $\ubox p \vee q$, as needed.

For the second, assume
\begin{enumerate*}[label=\arabic*)]
\item\label{ItDerivOneb} $\AxBInd p{\diam q}$ and
\item\label{ItDerivThreeb} $\ubox ( p\vee q)$.
\end{enumerate*}
From $\ubox(p\vee q)$, axiom \ref{ax10DiamFix} and some modal reasoning, $\ubox ( p\vee \diam q)$.
Also from axiom \ref{ax10DiamFix} and rule \ref{ax14NecCirc}, $\ubox(\tnext \diam q \to \diam q)$.
From $\AxBInd p{\diam q}$ we obtain $\ubox p \vee \diam q$, as needed.
\end{proof}

Finally, we introduce the {\em conditional excluded middle} axiom
\[{\rm CEM}  (p,q) \eqdef ( \neg \tnext p \wedge \tnext \neg \neg p ) \rightarrow ( \tnext q \vee \neg \tnext q ).\] 
This axiom states that a certain instance of excluded middle holds, provided some assumptions are satisfied.
It is less familiar than others we have considered, but its role will become clear when we consider semantics based on the real line.
With this, we define a handful of logics, listed in Table \ref{tableLogics}, along with definitions of the `optional' axioms.
The inclusions between these logics are summarized in Figure \ref{fig:fig}; as we will show in this paper, these are the only inclusions that hold between these logics.\footnote{
Note that our notation for logics has been modified from that in \cite{BoudouJelia}, in order to accommodate the larger family we now consider.
Specifically, $\logbasic$ was denoted ${\sf ITL}^0$, $\loghomeo$ was denoted $\sf ITL^{FS}$, $\logexp$ was denoted $\sf ITL^{CD}$, and $ \logpers $ was denoted ${\sf ITL}^1$.
Note that $\wlogbasic$ is weaker than ${\sf ITL}^0$.}
\begin{table}[H]
\begin{tabular}{rcl}
\hline
$\wboxfix (\varphi)$&$=$&$ \ubox\varphi \to \ubox \tnext \varphi$\\
${\rm FS}_\mtnext (\varphi,\psi)  $&$=$&$ \left(\tnext \varphi \rightarrow \tnext \psi \right) \to \tnext \left(\varphi \rightarrow \psi\right)$\\
$  {\rm FS}_\diam (\varphi,\psi) $ &$=$& $ \left( \diam \varphi \rightarrow \ubox \psi \right) \to \ubox \left(\varphi \rightarrow \psi\right)$\\
$\AxCons \varphi\psi  $&$=$&$ \ubox( \varphi \vee \psi) \to \ubox \varphi \vee \diam \psi$\\
$\AxConsm \varphi   $&$=$&$ \ubox( \neg \varphi \vee \varphi) \to \ubox \neg \varphi \vee \diam \varphi$\\
$\AxBInd \varphi\psi  $&$=$&$ \ubox ( \varphi \vee \psi) \wedge \ubox (\tnext \psi\rightarrow \psi) \rightarrow \ubox \varphi \vee \psi $\\
${\rm CEM}  (\varphi,\psi) $&$=$&$ ( \neg \tnext \varphi \wedge \tnext \neg \neg \varphi ) \rightarrow ( \tnext \psi \vee \neg \tnext \psi)$\\
\hline
\end{tabular}

\begin{tabular}{rclcrcl}
  $\logbasic$ & $=$ & \text{(See Definition \ref{defLogbasic})}  \\
    $\wlogbasic$ & $=$ & $\logbasic -  \ref{ax09BoxFix} + {\rm WH}$ & \ \ \ \  & $ \logbasica$ & $=$  &$ t_{\khence} (\wlogbasic)$ \\
  $\logrnax$ &$=$& $ \logbasic + {\rm CD}^- $ & \ \ \ \  & $ \logrnaxa$ & $=$ &$ t_{\khence} (\wlogbasic + {\rm CD}^-)$ \\
    $  \logrealax$&$=$&$   \logbasic + {\rm CD}^-  + {\rm CEM}$ &&    $   \logrealaxa $&$ = $&$ t_{\khence}(\wlogbasic + {\rm CD}^-  + {\rm CEM}) $ \\
$              \logexp $  &$=$&$   \logbasic + {\rm CD}$ &&              $ \logexpa $ & $= $& $t_{\khence}(\logexp) $\\
 $       \loghomeo $&$=$&$  \logbasic + {\rm FS}_\mtnext $ &&     $    \loghomeoa$ &$ = $&$ t_{\khence}(\loghomeo) $ \\
  $  \logrnfs$&$=$& $ \logbasic + {\rm FS}_\mtnext + {\rm CD}^- $ &&  $  \logrnfsa $&$ =$ & $t_{\khence}(\logrnfs) $ \\
  $\logpers  $ &$=$& $   \logbasic + {\rm FS}_\mtnext +  {\rm CD}$&&  $ \logpersa$ &$ = $&$ t_{\khence}(\logpers)$\\
  \hline
\end{tabular}
\caption{
Axioms not listed in Definition \ref{defLogbasic} (above) and logics based on strong and weak henceforth (below).
In the right-hand column, notice that only $\logbasica$, $\logrnaxa$ and $\logrealaxa$ are based on $\wlogbasic$.}
\label{tableLogics}
\end{table}
Here, $\logrealax$ stands for `real temporal logic', $\logrnax$ for `Euclidean temporal logic' and $\logexp$ for `constant domain temporal logic'.
For a logic $\Lambda$ in the above list, $ \Lambda^0$ is defined analogously but replacing $\logbasic$ by $\wlogbasic$.
Logics over $\mathcal L_{\diam{\khence}}$ are defined in Table \ref{tableLogics}.
Note that logics with either $\rm CD$ or ${\rm FS}_\mtnext$ use the strong axiom, ${\khence} \varphi \to \tnext {\khence}\varphi$.
This has to do with our semantics and will become clear later.

\begin{figure}[H]\centering
	\begin{tikzpicture}[thick,->,auto,font=\small,node distance=1.5cm]
	\node[] (logbasic) {$\logbasic$};
	\node[] (logrnax) 	[above left of=logbasic] {$\logrnax$};	
	\node[] (logrealax) [above right of=logrnax] {$\logrealax$};	
	\node[] (logexp) [above left of=logrnax] {$\logexp$};
	\node[] (logrnfs)  [above right of=logrealax] {$\logrnfs$};		
	\node[] (loghomeo)  [below right of=logrnfs] {$\loghomeo$};	
	\node[] (logpers) [above left of=logrnfs] {$\logpers$};
	
	\path[] 
	(logbasic) edge[] node[pos=0.5,right,font=\scriptsize]{}(loghomeo)
	(logbasic) edge[] node[pos=0.5,left,font=\scriptsize]{}(logrnax)	
	(loghomeo) edge[] node[pos=0.5,right,font=\scriptsize]{}(logrnfs)		
	(logexp) edge[] node[pos=0.5,right,font=\scriptsize]{}(logpers)	
	(logrealax) edge[] node[pos=0.5,right,font=\scriptsize]{}(logrnfs)		
	(logrnax) edge[] node[right,right,pos=0.5,font=\scriptsize]{}(logexp)
	(logrnax) edge[] node[right,pos=0.5,font=\scriptsize]{}(logrealax)		
	(logrnfs) edge[] node[right,pos=0.5,font=\scriptsize]{}(logpers)	
	;								
	\end{tikzpicture}
\hfill		\begin{tikzpicture}[thick,->,auto,font=\small,node distance=1.5cm]
	\node[] (logbasic) {$\logbasica$};
	\node[] (logrnax) 	[above left of=logbasic] {$\logrnaxa$};	
	\node[] (logrealax) [above right of=logrnax] {$\logrealaxa$};	
	\node[] (logexp) [above left of=logrnax] {$\logexpa$};
	\node[] (logrnfs)  [above right of=logrealax] {$\logrnfsa$};		
	\node[] (loghomeo)  [below right of=logrnfs] {$\loghomeoa$};	
	\node[] (logpers) [above left of=logrnfs] {$\logpersa$};
	
	\path[] 
	(logbasic) edge[] node[pos=0.5,right,font=\scriptsize]{}(loghomeo)
	(logbasic) edge[] node[pos=0.5,left,font=\scriptsize]{}(logrnax)	
	(loghomeo) edge[] node[pos=0.5,right,font=\scriptsize]{}(logrnfs)		
	(logexp) edge[] node[pos=0.5,right,font=\scriptsize]{}(logpers)	
	(logrealax) edge[] node[pos=0.5,right,font=\scriptsize]{}(logrnfs)		
	(logrnax) edge[] node[right,right,pos=0.5,font=\scriptsize]{}(logexp)
	(logrnax) edge[] node[right,pos=0.5,font=\scriptsize]{}(logrealax)		
	(logrnfs) edge[] node[right,pos=0.5,font=\scriptsize]{}(logpers)	
	;								
	\end{tikzpicture}
	\caption{Inclusions between the logics based on strong or weak henceforth we have defined; an arrow $\Lambda_1\to\Lambda_2$ means that every theorem of $\Lambda_1$ is a theorem of $\Lambda_2$.
}
	\label{fig:fig}
\end{figure}
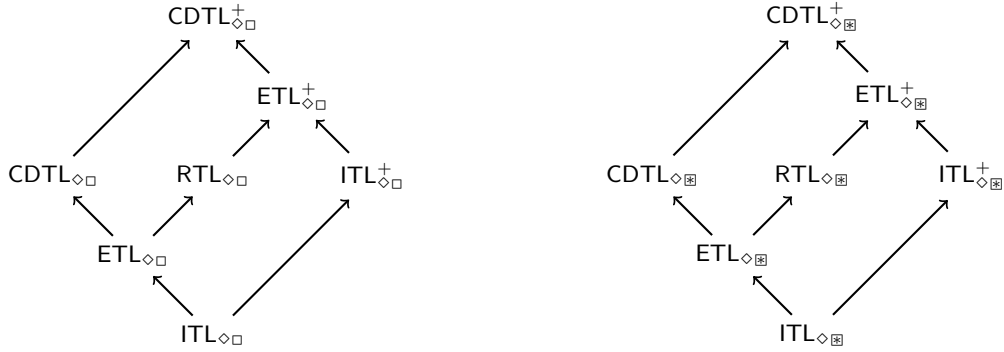	

Our list is not meant to exhaust all combinations of axioms; rather, we only consider logics that arise from natural classes of models.
Before discussing semantics, we establish the only non-trivial inclusion between these logics.

\begin{lemma}
Every instance of $\rm CEM$ is derivable in $\loghomeo$.
\end{lemma}

\proof
It is not hard to check that $\neg ( \neg \tnext p \wedge \tnext \neg \neg p )$ is derivable in $\loghomeo$, hence so is ${\rm CEM} $.
\endproof

This immediately yields that $\logrealax   \subseteq \logrnfs  $:

\begin{proposition}\label{propInclusions}
Every formula derivable in  $\logrealax $ is derivable in $\logrnfs  $.
\end{proposition}

We are also interested in logics over sublanguages of $\mathcal L_{\diam\ubox}$ or $\mathcal L_{\diam\khence}$.
For any logic $\Lambda $ defined above, let $\Lambda_\ubox$ be defined by restricting similarly all rules and axioms to $\mathcal L_\ubox$, except that when $\rm CD$ is an axiom of $\Lambda$, we add the axiom $\rm BI$ to $\Lambda_\ubox$.
In these cases, $\Lambda_{\khence}$ is similarly defined using $t_{\khence}({\rm BI})$.
The logic $\logbasicb$ is similar to a Hilbert calculus for the $\wedge,\vee$-free fragment considered by Yuse and Igarashi \cite{Yuse2006}, although they do not include induction but include the axioms $\ubox \varphi \to \ubox \ubox \varphi$ and $\ubox\tnext \varphi \leftrightarrow \tnext \ubox \varphi$.
It is not difficult to check that the latter are derivable from our basic axioms, and hence their logic is contained in $\logbasicb$.

We also define $\Lambda_\diam$ to be the logic obtained by restricting all rules and axioms to $\mathcal L_\diam$, and adding the rules $\frac{\varphi \to \psi}{\diam \varphi \to \diam \psi}$ and $\frac{\tnext \varphi \to \varphi}{\diam \varphi \to \varphi}$. Note that these rules correspond to axioms \ref{ax07:K:Dual}, \ref{ax12:ind:2}, respectively, but do not involve $\ubox$.
In this paper we are mostly concerned with logics including `henceforth', but $\ubox$-free logics are studied in detail by \cite{DieguezCompleteness}.

\section{Dynamic topological systems}\label{SecTopre}

The logics defined above are pairwise distinct.
We will show this by introducing semantics for each of them.
They will be based on dynamic topological systems (or dynamical systems for short), which, as was observed in \cite{FernandezITLc}, generalize their Kripke semantics \cite{BoudouCSL}.
In this section, we review the basic notions of topological dynamics needed in the rest of the text.
Let us first recall the definition of a {\em topological space,} as in e.g.~\cite{Dugundji}:

\begin{definition}
A {\em topological space} is a pair $\< X ,\mathcal{T}\>,$ where $X$ is a set and $\mathcal T$ a family of subsets of $X$ satisfying
\begin{enumerate*}[label=\alph*)]
	\item \mbox{$\varnothing,X\in \mathcal T$;}
	\item \mbox{if $U,V\in \mathcal T$ then $U\cap V\in \mathcal T$, and}
	\item \mbox{if $\mathcal O\subseteq\mathcal T$ then $\bigcup\mathcal O\in\mathcal T$.}
\end{enumerate*}
The elements of $\mathcal T$ are called {\em open sets}.
\end{definition}

If $x \in X$, a {\em neighbourhood} of $x$ is an open set $U \subseteq X$ such that $x \in U$. Given a set $A\subseteq X$, its {\em interior}, denoted $A^\circ$, is the largest open set contained in $A$.
It is defined formally by
\begin{equation}\label{EqInterior}
A^\circ=\bigcup\cbra U\in\mathcal T:U\subseteq A\cket.
\end{equation}
Dually, we define the closure $\overline A$ as $X\setminus(X\setminus A)^\circ$; this is the smallest closed set containing $A$.

If $\<X,\mathcal T\>$ is a topological space, a function $S\colon X \to X$ is {\em continuous} if, whenever $U \subseteq X$ is open, it follows that $S^{-1}[U]$ is open.
The function $S$ is {\em open} if, whenever $V \subseteq X$ is open, then so is $S[V]$.
An open, continuous function is an {\em interior map}, and a bijective interior map is a {\em homeomorphism;} equivalently, $S$ is a homeomorphism if it is invertible and both $S$ and $S^{-1}$ are continuous.

A dynamical system is then a topological space equipped with a continuous function:

\begin{definition}
A {\em dynamical (topological) system} is a triple $\mathcal X = (X,\mathcal T,S)$ such that $(X,\mathcal T)$ is a topological space and $S\colon X \to X$ is continuous. We say that $\mathcal X$ is {\em open} if $S$ is an interior map and {\em invertible} if $S$ is a homeomorphism.
\end{definition}
%

Topological spaces generalize posets in the following way. Let $\mathcal F=\<W,{\acc}\>$ be a poset; that is, $W$ is any set and $\acc$ is a transitive, reflexive, antisymmetric relation on $W$.
To see $\mathcal F$ as a topological space, define $\mathord \uparrow w=\cbra v:w\peq v\cket.$
Then consider the topology $\mathcal T_\peq$ on $W$ given by setting $U\subseteq W$ to be open if and only if, whenever $w\in U$, we have $\mathord \uparrow w\subseteq U$. A topology of this form is an {\em up-set topology} \cite{alek}.
The interior operator on such a topological space is given by
\begin{equation}\label{EqIntPoset}
A^\circ = \{w \in W : {\uparrow} w\subseteq A\};
\end{equation}
i.e., $w$ lies on the interior of $A$ if whenever $v \seq w$, it follows that $v\in A$.

Throughout this text we will often identify partial orders with their corresponding topologies, and many times do so tacitly.
In particular, a dynamical system generated by a poset is called a \emph{dynamical} or \emph{expanding poset}, due to its relation to expanding products of modal logics \cite{GabelaiaExpanding}.
It will be useful to characterize the continuous and open functions on posets (see Figure~\ref{fig:continuous-open}):
\begin{lemma}
  Consider a poset $\< W, \mathord\peq \>$ and a function $S\colon W \to W$. Then,
  \begin{enumerate}
  
\item   the function $S$ is continuous with respect to the up-set topology if and only if, whenever $w \peq w'$, it follows that $S(w) \peq S(w')$, and
  
\item   the function $S$ is open with respect to the up-set topology if whenever $S(w) \peq v$, there is $w' \in W$ such that $w \peq w'$ and $S(w') = v$.
    
  \end{enumerate}
\end{lemma}
These are properties common in multi-modal logics and we refer to them as `confluence properties.'
A {\em persistent} function is an open, continuous map on a poset.

\begin{figure}[H]\centering
\begin{subfigure}{.2\textwidth}
\begin{tikzpicture}[thick,->,auto,font=\small,node distance=1.5cm]
\node[state,minimum size=3.5mm,text width=3mm] (n1) {$w$};
\node[state,minimum size=3.5mm,text width=3mm] (n2) [right of=n1] {\phantom{w}};
\node[state,minimum size=3.5mm,text width=3mm] (n3) [above of=n2] {\phantom{w}};
\node[state,minimum size=3.5mm,text width=3mm] (n4) [above of=n1] {$w'$};

\path[->,very thick] (n1) edge node{$S$}(n2)
					 (n1) edge node{$\peq$}(n4)	
					(n4) edge node{$S$}(n3)
					(n2) edge[dashed] node[right] {$\peq$} (n3);

\end{tikzpicture}
\caption{Continuity}
\end{subfigure}\hspace{100pt}\begin{subfigure}{.2\textwidth}
\begin{tikzpicture}[thick,->,auto,font=\small,node distance=1.5cm]
\node[state,minimum size=3.5mm,text width=3mm] (n1) {$w$};
\node[state,minimum size=3.5mm,text width=3mm] (n2) [right of=n1] {\phantom{w}};
\node[state,minimum size=3.5mm,text width=3mm] (n3) [above of=n2] {$v$};
\node[state,minimum size=3.5mm,text width=3mm] (n4) [above of=n1] {$w'$};

\path[->,very thick] (n1) edge node{$S$}(n2)
(n1) edge[dashed] node{$\peq$}(n4)	
(n4) edge[dashed] node{$S$}(n3)
(n2) edge[] node[right] {$\peq$} (n3);
\end{tikzpicture}
\caption{Openness}
\end{subfigure}
\caption{On a dynamic poset, the above diagrams can always be completed if $S$ is continuous or open, respectively. }\label{FigCO}
\label{fig:continuous-open}
%
%
%
%
%
%
%
%
%
%
%
%
%
%
%
%
%
%
%
%
%
%
%
%
%
%
%
%
%
%
%
%
%
%
%
%
%
%
%
%
%
%
%
%
%


\end{figure}
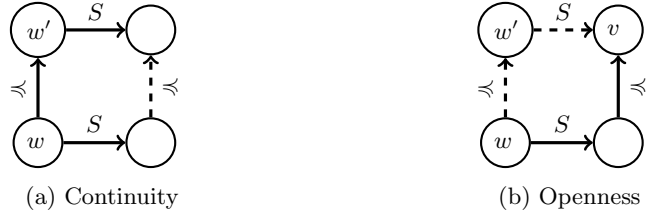

\section{Semantics}\label{SecSemantics}

In this section we will see how dynamical systems can be used to provide a natural intuitionistic semantics for the language of linear temporal logic.
Classicall $\sf LTL$ may be interpreted over structures $(X,S)$ where $X$ is a set and $S\colon X\to X$.
In this setting, $\tnext\varphi$ is true on a point $x$ if $\varphi$ is true `at the next moment,' i.e., on $S(x)$; $\diam\varphi$ is true on a point $x$ if $\varphi$ is `eventually' true, i.e.~there is $n\geq 0$ such that $\varphi$ is true on $S^n(x)$, and $\ubox \varphi$ is true on a point $x$ if $\varphi$ is `henceforth' true, i.e.~for all $n\geq 0$, $\varphi$ is true on $S^n(x)$.
Meanwhile, topological spaces provide semantics for intuitionistic logic, where each formula is assigned an open set.
Under this semantics, $\varphi\to\psi$ is true on $x$ if there is a neighborhood $U$ of $x$ (i.e., an open set $U$ with $x\in U$) such that every $y\in U$ satisfying $\varphi$ also satisfies $\psi$.
Thus it is natural to interpret intuitionistic temporal logic on dynamical systems, which are endowed with both a topology and a transition function.
In this setting, the classical definitions of $\tnext$ and $\diam$ readily adapt to the topological setting without modification.
On the other hand, the classical definition of $\ubox$ does not necessarily yield open sets, and to this end we consider two variants of henceforth, the weak variant, $\khence$, and the strong variant, which we simply denote $\ubox$.

\begin{definition}\label{DefSem}
Given a dynamical system $\mathcal X=(X,\mathcal T ,S)$, a {\em valuation on $\mathcal X$} is a function $\lb\cdot\rb\colon\mathcal L_{\diam\ubox{\khence}}\to \mathcal T$ such that
{\setlength{\columnsep}{-1cm}\begin{multicols}{2}
\begin{itemize}
\item[]$\lb\bot\rb=\varnothing$
\item[]$\lb\varphi\wedge\psi\rb=\lb\varphi\rb\cap \lb\psi\rb$
\item[]$\lb\varphi\vee\psi\rb =\lb\varphi\rb\cup \lb\psi\rb$
\item[]$\lb\varphi\to\psi\rb = \big ( (X\setminus\lb\varphi\rb)\cup \lb\psi\rb\big )^\circ$ 	
\item[]$\val{\tnext\varphi}=S^{-1} \val\varphi$
\item[]$\val{\diam\varphi}=\textstyle\bigcup\limits_{n\geq 0}S^{-n} \val\varphi$
\item[]$\val{\ubox\varphi}= \bigcup \ \Big \{U\in \mathcal T : S[U] \subseteq U\subseteq \val\varphi \Big \}$
\item[]$\val{{\khence}\varphi}= \Big (  \textstyle\bigcap\limits_{n\geq 0} S^{-n} \val \varphi \Big )^\circ $
\end{itemize}
\end{multicols}
}
A tuple $\mathcal M = (X, \mathcal T, S, \val\cdot)$ consisting of a dynamical system with a valuation is a {\em dynamic topological model,} and if $\mathcal T$ is generated by a partial order, we will say that $\mathcal M$ is a {\em dynamic poset model.}
\end{definition}

All of the semantic clauses are standard from either intuitionistic or temporal logic, with the exception of those for ${\khence}\varphi$ and $\ubox\varphi$, which we discuss in greater detail in the remainder of this section. It is not hard to check by structural induction on $\varphi$ that $\val \varphi$ is uniquely defined given any assignment of the propositional variables to open sets, and that $\val \varphi$ is always open. We define validity in the standard way, and with this introduce additional semantically-defined logics, two of which were studied in \cite{BoudouCSL}.

\begin{definition}\label{defValidForms}
If $\mathcal M = (X,\mathcal T,S,\val\cdot)$ is any dynamic topological model and $\varphi \in \mathcal L$ is any formula, we write $\mathcal M \models \varphi$ if $\val\varphi =X$. Similarly, if $\mathcal S= (X,\mathcal T,S)$ is a dynamical system, we write $\mathcal S \models \varphi$ if for any valuation $\val \cdot$ on $\mathcal X$, we have that $(\mathcal S,\val\cdot) \models \varphi$; and, if $\mathcal X = (X,\mathcal T)$ is any topological space, then $\mathcal X \models \varphi$ if, for any continuous $S\colon X\to X$, $(X,\mathcal T,S) \models\varphi$.
Finally, if $\Omega$ is a class of structures (either topological spaces, dynamical systems, or models), we write $\Omega \models \varphi$ if for every $\mathcal A \in \Omega$, $\mathcal A \models \varphi$, in which case we say that $\varphi$ is {\em valid} on $\Omega$.

If $\Omega$ is either a structure or a class of structures and $\Theta\subseteq \{\diam,\ubox,{\khence}\}$, we write ${\sf ITL}^\Omega_\Theta$ for the set of $\mathcal L _\Theta$ formulas valid on $\Omega$.
We maintain the convention that $\tnext$ is assumed to be in all languages, and we write ${\sf ITL}^\Omega_\mtnext$ when $\Theta = \varnothing$.
\end{definition}

The main classes of dynamical systems we are interested in are listed in Table \ref{tableClasses}.
For example, $\itlc$ denotes the set of $\mathcal L_{\diam\ubox}$-formulas valid over the class of {\em all} dynamical systems.

\begin{table}[H]
\begin{tabular}{cl}
\hline
$\sf c$&all dynamical systems (with a continuous function)\\
$\sf e$&expanding posets\\
$\sf p$&persistent posets\\
$\sf o $&open dynamical systems\\
$\mathbb R^n$ &systems based on $n$-dimensional Euclidean space\\
\hline
\end{tabular}
\caption{The main classes of dynamical systems appearing in the text.}
\label{tableClasses}
\end{table}

In practice, it is convenient to have a `pointwise' characterization of the semantic clauses of Definition \ref{DefSem}. For a model $\mathcal M = (X, \mathcal T, S, \val\cdot)$, $x \in X$ and $\varphi \in \mathcal L$, we write $\mathcal M,x \models \varphi$ if $x\in \val \varphi$, and $\mathcal M \models \varphi $ if $\val \varphi = X$.
Then, in view of \eqref{EqInterior}, given formulas $\varphi$ and $\psi$, we have that $\mathcal M, x \models {\varphi \to \psi}$ if and only if there is a neighbourhood $U$ of $x$ such that for all $y \in U$, if $\mathcal M, y \models \varphi$ then $\mathcal M, y \ \models \psi$; note that this is a special case of {\em neighbourhood semantics} \cite{PacuitNeighborhood}.
    The following simple observation will be useful.

\begin{lemma}\label{LemImpCrit} If $\mathcal M = (X, \mathcal T, S, \val\cdot)$ is any model and $\varphi,\psi \in \mathcal L_{\diam\ubox\khence}$, then $\mathcal M \models \varphi \to \psi$ if and only if $\lb \varphi\rb \subseteq \lb \psi \rb$.   
\end{lemma}

\begin{proof}
  	 If $\lb \varphi\rb \subseteq \lb \psi \rb$ then $(X\setminus\lb\varphi\rb)\cup \lb\psi\rb =  X$, so
  	$\lb\varphi\to\psi\rb =\big ( (X\setminus\lb\varphi\rb)\cup \lb\psi\rb\big )^\circ = X^\circ = X .$
  	Otherwise, there is  $z \in \lb \varphi\rb $ such that $z \notin \lb \psi\rb$, so that $ z \notin \big ( (X\setminus\lb\varphi\rb)\cup \lb\psi\rb\big )^\circ$, i.e.~$z \notin  \lb\varphi \rightarrow \psi \rb $.
\end{proof}

Using \eqref{EqIntPoset}, this can be simplified somewhat in the case that $\mathcal T$ is generated by a partial order $\peq$:

\begin{proposition}
If $(X,{\peq},S,\val\cdot)$ is a dynamic poset model, $x\in X$, and $\varphi$, $\psi$ are formulas, then $\mathcal M, x \models {\varphi \to \psi}$ if and only if from $y \seq x$ and $\mathcal M, y \models \varphi$, it follows that $\mathcal M, y \models \psi$.
\end{proposition}

\noindent This is the standard relational interpretation of implication, and thus topological semantics are a generalization of the usual Kripke semantics.

The semantics for ${\khence}$ were originally introduced by \cite{KremerIntuitionistic} as an intuitionistic reading of `henceforth'. By analogy with $\diam$, one might try to interpret $\val{{\khence} \varphi}$ as $\bigcap_{n\geq 0} \ S^{-n} \val\varphi$. But this does not quite work since, on this interpretation, there would be no guarantee that $\val{{\khence}\varphi}$ is open. Instead, we consider interpreting $\val{{\khence}\varphi}$ as the \emph{interior} of $\bigcap_{n\geq 0} \ S^{-n} \val\varphi$.
In other words, $\mathcal M, x \models {\khence} \varphi $ if and only if there is a neighbourhood $U$ of $x$ so that for every $y\in U$ and every $n\in \mathbb N$, one has that $\mathcal M, S^n(y) \models \varphi $.

This interpretation of `henceforth' is analogous to the interpretation in \cite{RasiowaMeta}, and going back to \cite{Mostowski}, of $\forall x$ in the topological semantics for quantified intuitionistic logic. We may interpret variables as ranging over some non-empty domain $D$, and truth values as open sets in some topological space $(X,\mathcal T)$.
The semantic clauses for $\exists x$ and $\forall x$ are, essentially, the following:
\begin{eqnarray*}
\val{\exists x \varphi} & = & \textstyle\bigcup\limits_{d \in D} \val{\varphi[d/x]} \\
\val{\forall x \varphi} & = & (\textstyle\bigcap\limits_{d \in D} \val{\varphi[d/x]})^\circ
\end{eqnarray*}
Note that if $D$ is infinite then the intersection in the definition of $ \val{\forall x \varphi}$ may also be infinite and hence the application of the interior operator is necessary in order to obtain an open truth value.

The semantics for $\ubox\varphi$ are also an intuitionistic interpretation of `henceforth', but from a more algebraic perspective.
In classical temporal logic, $\lb \ubox\varphi\rb$ is the largest set contained in $\lb\varphi\rb$ which is closed under $S$. In our semantics, $\lb\ubox\varphi\rb$ is the greatest {\em open} set which is closed under $S$. 
If $\mathcal M, x \models {\ubox \varphi}$, this fact is witnessed by an open, $S$\mbox{-}invariant neighbourhood of $x$, where $U\subseteq X$ is {\em $S$-invariant} if $S[U] \subseteq U$.

\begin{proposition}
If $(X,{\mathcal T},S,\val\cdot)$ is a dynamic topological model, $x\in X$, and $\varphi$ is any formula, then $\mathcal M, x \models {\ubox \varphi}$ if and only if there is an $S$-invariant neighbourhood $U$ of $x$ such that for all $y \in U$, $\mathcal M, y \models \varphi$.
\end{proposition}

In fact, the open, $S$-invariant sets form a topology; that is, the family of $S$-invariant open sets is closed under finite intersections and arbitrary unions, and both the empty set and the full space are open and $S$-invariant (this  follows readily from the fact that the topology $\mathcal T $ already has these properties, as does the family of $S$-invariant sets). This topology is coarser than $\mathcal T$, in the sense that every $S$-invariant open set is (tautologically) open. Thus $\ubox$ can itself be seen as an interior operator based on a coarsening of $\mathcal T$, and $\val{\ubox\varphi}$ is always an $S$-invariant open set.

\begin{example}\label{ExBoxOnR}
As usual, the real number line is denoted by $\mathbb R$ and we assume that it is equipped with the standard topology, where $U \subseteq \mathbb R$ is open if and only if it is a union of intervals of the form $(a,b)$.
Consider a dynamical system based on $\mathbb R$ with $S \colon \mathbb R \to \mathbb R$ given by $S(x) = 2x$.
We claim that for any model $\mathcal M$ based on $(\mathbb R, S)$ and any formula $\varphi$, $\mathcal M, 0 \models {\ubox \varphi}$ if and only if $ \mathcal M \models \varphi$.

To see this, note that one implication is obvious since $\mathbb R$ is open and $S$-invariant, so if $\val\varphi = \mathbb R$ it follows that $ \mathcal M, 0 \models {\ubox\varphi}$.
For the other implication, assume that $\mathcal M, 0 \models {\ubox\varphi}$, so that there is an $S$-invariant, open $U\subseteq \val\varphi$ with $0 \in U$.
It follows from $U$ being open that for some $\varepsilon > 0$, $(-\varepsilon,\varepsilon) \subseteq U$.
Now, let $x \in \mathbb R$, and let $n$ be large enough so that $|2^{-n} x| < \varepsilon$.
Then, $2^{-n} x \in U$, and since $U$ is $S$-invariant, $x = S^n (2^{-n} x ) \in U$.
Since $x$ was arbitrary, $U = \mathbb R$, and it follows that $\mathcal M \models \varphi$.

On the other hand, suppose that $ 0 < a <x$ and $(a,\infty) \subseteq \val\varphi$.
Then, $(a,\infty)$ is open and $S$-invariant, so it follows that $x \in \val{\ubox\varphi}$.
Hence in this case we do not require that $\val\varphi = \mathbb R$. Similarly, if $x<a<0$ and $(-\infty,a) \subseteq \val \varphi$, we readily obtain $x \in \val{\ubox\varphi}$.
\end{example}
%

We will see more examples in Section \ref{SecInd}, where we show, among other things, that the two interpretations of `henceforth' are not equivalent.
In general, we only obtain one implication.

\begin{lemma}\label{lemmStrongtoWeak}
The formula $\ubox p \to {\khence} p$ is valid over the class of dynamical systems.
\end{lemma}

\begin{proof}
Let $\mathcal M = (X,\mathcal T,S,\val\cdot)$ be any dynamical model.
Suppose that $w\in \val {\ubox p}$, and let $U$ be an $S$-invariant neighbourhood of $x$ such that $U\subseteq \val  p$.
Then, using the $S$-invariance of $U$ we see by a routine induction on $n$ that $U \subseteq S^{-n}\val p $, hence $U\subseteq \bigcap _{n\in \mathbb N}S^{-n} \val p  $.
As $U$ is open,  $U\subseteq \left ( \bigcap _{n\in \mathbb N}S^{-n} \val \varphi \right )^\circ $, and hence $w\in \val{{\khence} p }$.
Since $w$ was arbitrary, $\mathcal M \models \ubox p \to {\khence} p$.
\end{proof}

However, when restricted to `nice' dynamical systems, the two versions of `henceforth' coincide.

\begin{proposition}\label{propBoxKripke}
Let $\mathcal M = (W,{\peq},S,\val \cdot)$ be any dynamic poset model, $w\in W$ and $\varphi \in \mathcal L$. Then, the following are equivalent:

\begin{multicols}3
\begin{enumerate}[label=\alph*)]
	\item\label{ItOfficial}\mbox{$w\in \val{\ubox\varphi}$;}
	\columnbreak
	\item\label{ItInterior}\mbox{$w\in \val{{\khence}\varphi}$;}
	\columnbreak
	\item\label{ItKripke} \mbox{$\forall n \in \mathbb N \ \big ( S^n (w)  \in \val \varphi \big )$.}
\end{enumerate}
\end{multicols}

\end{proposition}

\begin{proof}
By Lemma \ref{lemmStrongtoWeak}, \ref{ItOfficial} implies \ref{ItInterior}.
That \ref{ItInterior} implies \ref{ItKripke} is immediate from
$\left ( \bigcap _{n\in\mathbb N} S^{-n} \val \varphi \right )^\circ \subseteq \bigcap _{n\in\mathbb N} S^{-n} \val \varphi ,$
so it remains to show that \ref{ItKripke} implies \ref{ItOfficial}.
Suppose that for all $n\in \mathbb N$, $\mathcal M, S^n (w)  \models \varphi$, and let $U = \bigcup_{n \in \mathbb N} {\uparrow} S^n(w)$. That the set $U$ is open follows from each ${\uparrow} S^n(w)$ being open and unions of opens being open.
If $v \in U$, then $v \seq S^n(w)$ for some $n\in \mathbb N$ and hence by upwards persistence, from $\mathcal M, S^n(w) \models \varphi$ we obtain $\mathcal M, v\models \varphi$; moreover, $S(v) \seq S^{n+1}(w)$ so $S(v) \in U$. Since $v \in U$ was arbitrary, we conclude that $U$ is $S$-invariant and $U\subseteq \val \varphi$.
Thus $U$ witnesses that $\mathcal M, w\models \ubox \varphi$.
\end{proof}

As we will see later, Proposition \ref{propBoxKripke} fails over the class of general dynamical systems, but a weaker version holds over the class of open dynamical systems.

\begin{proposition}\label{propBoxHomeo}
The formula ${\khence} p \leftrightarrow \ubox p$ is valid over the class of open dynamical systems.
\end{proposition}

\proof
One implication is Lemma \ref{lemmStrongtoWeak}, so we focus on the other.
Let $(X,\mathcal T,S,\val\cdot)$ be a dynamical model.
Assume that $w\in \val{{\khence} p}$, and let $U = \left ( \bigcap _{n\in \mathbb N} S^{-n} \val p \right )^\circ$.
Clearly $U$ is open; we claim that it is $S$-invariant.
We have that
\[S[U] = S \left ( \left ( \bigcap _{n\in \mathbb N} S^{-n} \val p \right )^\circ \right ) \subseteq S  \left ( \bigcap _{n\in \mathbb N} S^{-n} \val p \right ) \subseteq   \bigcap _{n\in \mathbb N} S^{-n} \val p , \]
where the latter inclusion is obtained by distributing $S$ over the intersection.
Moreover, $S[U]$ is open, since $S$ is assumed to be an open function.
Thus
\[S[U] \subseteq \left ( \bigcap _{n\in \mathbb N} S^{-n} \val p \right )^\circ = U,\]
witnessing that $w\in \val{\ubox p}$.
\endproof

\section{Soundness}\label{SecSound}

In this section we will show that several of the logics we have considered are sound for their semantics based on different classes of dynamic topological systems. 
    First we show that our basic logics (as given in Definition \ref{defLogbasic} and Table \ref{tableLogics}) are sound for the class of {\em all} dynamical systems.
Below, recall that $\sf c$ denotes the class of all dynamical systems (see Table \ref{tableClasses}).
   
\begin{theorem}\label{ThmSoundZero}
The logics $\logbasic$ and $\logbasica$ are sound for the class of dynamical systems; that is, $\logbasic \subseteq \itlc$ and $\logbasica \subseteq \itlca$.
  \end{theorem}

\begin{proof}
Let $\mathcal M = (X , \mathcal T, S, \val \cdot)$ be any dynamical topological model; we must check that all the axioms \ref{ax01Taut}-\ref{ax12:ind:2} are valid on $\mathcal M$ and the rules \ref{ax13MP}, \ref{ax14NecCirc} preserve validity. Note that all intuitionistic tautologies are valid due to the soundness for topological semantics \cite{MintsInt}. Many of the other axioms can be checked routinely, so we focus only on those axioms involving the continuity of $S$ or the semantics for $\ubox$ or ${\khence}$.
\medskip

\ignore{
\noindent \ref{ax02Bot}-\ref{ax04NexVee} These axioms follow from the fact that $S^{-1}[\cdot]$ is a Boolean homomorphism; for example, for \ref{ax03NexWedge},
\begin{align*}
&\lb \tnext \left( \varphi \wedge \psi \right) \rb = S^{-1}\val{\varphi \wedge \psi} = S^{-1} \big [\val{\varphi} \cap \val{\psi} \big ]\\
& = S^{-1} \val{\psi} \cap  S^{-1} \val{\psi} = \val { \tnext \varphi } \cap \val{ \tnext \psi } = \lb \tnext \varphi \wedge\tnext \psi \rb,
\end{align*}
so $\mathcal M \models \tnext (\varphi \wedge \psi) \leftrightarrow \tnext \varphi \wedge \tnext \psi$ in view of Lemma \ref{LemImpCrit}. Henceforth, we will use the latter lemma without mention.
\medskip
}

\noindent{\sc Axiom \ref{ax05KNext}: $\tnext(\varphi\to\psi) \to(\tnext\varphi\to\tnext \psi)$.} Suppose that $x\in \val{\tnext(\varphi \to \psi)}$. Then, $S(x) \in \val {\varphi \to \psi}$. Since $S$ is continuous and $\val {\varphi \to \psi}$ is open, $U = S^{-1} \val{\varphi \to \psi}$ is a neighbourhood of $x$. Then, for $y \in U$, if $y \in \val{\tnext \varphi}$, it follows that $S(y) \in \val \varphi \cap \val {\varphi \to \psi}$, so that $S(y) \in \val \psi$ and $y \in \val {\tnext \psi}$. Since $y \in U$ was arbitrary, $x \in \val{\tnext\varphi \to \tnext \psi}$, thus $\val{\tnext(\varphi \to \psi)} \subseteq \val {\tnext \varphi \to \tnext \psi}$, and by Lemma \ref{LemImpCrit} (which we will henceforth use without mention), axiom \ref{ax05KNext} is valid on $\mathcal M$.
\medskip

\noindent{\sc Axiom \ref{ax06KBox}: $\ubox(\varphi\to\psi)\to (\ubox\varphi\to\ubox\psi)$.} Observe that $\val{\ubox(\varphi \to \psi)}$ is an $S$-invariant open subset of $\val{\varphi \to \psi}$.
Similarly, $\val{\ubox\varphi}$ is an $S$-invariant open subset of $\val{\varphi}$.
Let
$U = \val{\ubox(\varphi \to \psi)} \cap \val{\ubox\varphi}$.
Since $U$ is open, it suffices to prove that $U \subseteq \val{\ubox\psi}$.
Moreover, $U$ is $S$-invariant, therefore it suffices to prove that $U \subseteq \val{\psi}$,
which is direct because $U \subseteq \val{\varphi \to \psi} \cap \val{\varphi}$ and $\val{\varphi \to \psi} \subseteq (X \setminus \val{\varphi}) \cup \val{\psi}$.
\medskip

\noindent{\sc Axiom \ref{ax07:K:Dual}: $\ubox(\varphi\to\psi)\to (\diam\varphi\to \diam \psi)$.} As before, suppose that $x\in \val{\ubox(\varphi \to \psi)}$, and let $U$ be an $S$-invariant neighbourhood of $x$ such that $U \subseteq \val{\varphi \to \psi}$. If $y \in U \cap \val{\diam \varphi}$, then $S^n(y)\in \val \varphi$ for some $n$; since $U$ is $S$-invariant, $S^n(y) \in U$, hence $S^n(y) \in \val {\psi}$ and $y \in \val{\diam \psi}$. We conclude that $x \in \val {\diam \varphi \to \diam \psi}$.
\medskip

\noindent \noindent{\sc Axioms \ref{ax09BoxT}, \ref{ax09BoxFix}: $\ubox\varphi\to \varphi\wedge\tnext\ubox\varphi$.} Suppose that $x\in \val {\ubox\varphi}$, and let $U\subseteq \val \varphi$ be an $S$-invariant neighbourhood of $x$. Then, $x \in U$, so $x \in \val \varphi$. Moreover, $U$ is also an $S$-invariant neighbourhood of $S(x)$, so $S(x) \in \val {\ubox \varphi}$ and thus $x \in \val{\tnext \ubox \varphi}$. We conclude that $x \in \val{\varphi }\cap \val{ \tnext\ubox \varphi}$.
\medskip

\ignore{
\noindent \ref{ax10DiamFix} This axiom is also standard from relational semantics. If $x\in \val {\varphi \vee \tnext \diam \varphi}$, then either $x \in \val \varphi$ and hence $S^0(x) \in \val \varphi$, yielding $x\in \val {\diam \varphi}$, or $S(x) \in \val{\diam \varphi}$, yielding $S^{n+1}(x) = S^n(S(x)) \in \val \varphi$ for some $n$, hence $x\in \val {\diam \varphi}$.
\medskip
}

\noindent{\sc Axiom \ref{ax11:ind:1}: $\ubox (\varphi \to \tnext \varphi) \to (\varphi\to \ubox \varphi)$.} Suppose that $x\in \val{\ubox (\varphi \to \tnext \varphi)}$. If $x \in \val \varphi$, then $U = \val \varphi \cap \val{\ubox (\varphi \to \tnext \varphi)}$ is open (by the intuitionistic semantics) and $S$-invariant, since if $y \in U$, from $y \in \val{\varphi \to \tnext\varphi}$ we obtain $S(y) \in \val \varphi$. It follows that $U$ is an $S$-invariant neighbourhood of $x$, so $x\in \val {\ubox \varphi}$.
\medskip

\noindent{\sc Axiom \ref{ax12:ind:2}: $\ubox (\tnext \varphi \to \varphi) \to (\diam \varphi\to  \varphi)$.} Suppose that $x\in \val{ \ubox(\tnext \varphi \to \varphi)} \cap \val{ \diam \varphi}$. Let $U\subseteq \val {\tnext \varphi \to \varphi}$ be an $S$-invariant neighbourhood of $x$.
Let $n$ be least so that $S^n(x) \in \val \varphi$; if $n > 0$, since $U$ is $S$-invariant we see that $S^{n-1}(x) \in U\subseteq  \val {\tnext\varphi \to \varphi}$, hence $S^{n-1}(x) \in \val \varphi$, contradicting the minimality of $n$. Thus $n = 0$ and $x \in \val \varphi$.
\medskip


%

\noindent{\sc Axiom $ t_{\khence}$\ref{ax06KBox}: $\khence(\varphi\to\psi) \to (\khence\varphi\to\khence \psi)$.} Assume that $x\in \val{{\khence}(\varphi \to \psi)}$. We claim that if $y \in \val{{\khence}(\varphi \to \psi)} \cap \val {{\khence} \varphi}$ then $y \in \val{{\khence} \psi}$, from which we obtain $x \in \val{{\khence} \varphi \to {\khence}\psi}$.
If $y \in \val{{\khence}(\varphi \to \psi)} \cap \val{{\khence} \varphi}$ then note that for all $z \in  \val{{\khence}(\varphi \to \psi)} \cap \val{{\khence} \varphi}$, by definition of the semantics, we have that for all $k \ge 0$, $S^k(z) \in \val{\varphi \to \psi}$ and $S^k(z) \in \val{\varphi}$, so $S^k(z) \in \val{\psi}$. 
Since $U:=\val{{\khence}(\varphi \to \psi)} \cap \val{{\khence} \varphi}$ is a neighbourhood of $y$ and $z\in U$ was arbitrary,
$y \in \val{{\khence} \psi}$.
Since $O:=\val{{\khence}(\varphi \to \psi)}$ is a neighbourhood of $x$ and $y\in O$ was arbitrary, this witnesses that $x\in \val{{\khence} \varphi \to {\khence} \psi}$.

\medskip

\longversion{
\noindent $t_{\khence}$\ref{ax07:K:Dual}  suppose that $x \in \val{{\khence}(\varphi \to \psi)}$. We claim that if $y\in  \val{{\khence}(\varphi \to \psi)} \cap \val{\diam \varphi}$ then $y \in \val{\diam \psi}$, for all $y$. If $y\in  \val{{\khence}(\varphi \to \psi)} \cap \val{\diam \varphi}$ then $S^i(y) \in \val{\varphi} \cap \val{\varphi \to \psi}$ for some $i\ge 0$. Therefore,
$S^i(y) \in \val{\psi}$ and $y\in \val{\diam \psi}$. Since $U:= \val{{\khence}(\varphi \to \psi)} \cap \val{\diam \varphi}$ is a neighbourhood of $x$ and $y$ was arbitrary then $x \in \val{\diam\varphi \to \diam \psi}$.
\medskip
}

\noindent{\sc Axiom $t_{\khence}$(\wboxfix): $\khence\varphi\to\khence\tnext\varphi $.} Suppose that $x \in \val{{\khence}\varphi }$. Let $U:= \val{{\khence} \varphi}$ be a neighbourhood of $x$ such that, for all $y\in U$ and $n\in \mathbb N$, $S^n(y) \in \val \varphi$.
Therefore, for all $y\in U$ and $n\in \mathbb N$, $S^{n+1}(y) \in \val \varphi$, i.e., $y\in \val{\tnext\varphi}$, so that $U$ witnesses that $x \in \val{{\khence}\tnext\varphi }$.

\medskip

\noindent{\sc Axiom $t_{\khence}$\ref{ax11:ind:1}: ${\khence} (\varphi \to \tnext \varphi) \to (\varphi\to\khence\varphi)$.} Suppose that $x\in \val{{\khence} (\varphi \to \tnext \varphi)}$. If $x \in \val{\varphi}$, then $U = \val{\varphi} \cap \val{{\khence} (\varphi \to \tnext \varphi)}$ is a neighbourhood of $x$. It can be proved by induction on $i$ that for all $i \ge 0$ and $y\in U$, $S^i(y) \in \val{\psi}$, so that $U$ witnesses that $x \in \val{{\khence} \psi}$.
\medskip

\longversion{
\noindent $t_{\khence}$\ref{ax12:ind:2} Suppose that $x\in \val{ {\khence}(\tnext \varphi \to \varphi)}$. Since $x \in \val{\varphi}$, take the least $i > 0$ for which $S^i(x)  \in \val{\varphi}$. If $i>0$, since $x \in \val{{\khence}(\tnext \varphi \to \varphi)}$ then, $S^{i-1}(x) \in \val{\varphi}$ contradicting the mininality of $i$. Therefore, $i=0$ and $x\in \val{\varphi}$.
}
\noindent{\sc Axioms $t_{\khence}$\ref{ax07:K:Dual}, $t_{\khence}$\ref{ax12:ind:2}:} These are treated as their analogues for $\ubox$.
\end{proof}


The additional axioms we have considered are valid over specific classes of dynamical systems. Specifically, the constant domain axiom is valid for the class of expanding posets, while the Fischer Servi axioms are valid for the class of open systems. Let us begin by discussing the former in more detail.
In the next few results, recall that the relevant definitions are summarized in Tables \ref{tableLogics} and \ref{tableClasses}.

\begin{theorem}\label{ThmSoundCD}
The logics $\logexp$ and $\logexpb$ are sound for the class of expanding posets; that is, $\logexp \subseteq \itle$ and $\logexpb \subseteq \itleb$.
\end{theorem}

\begin{proof}	
Let $\mathcal M = (X,{\peq}, S, \val\cdot)$ be a dynamic poset model; in view of Theorem~\ref{ThmSoundZero}, it only remains to check that $\rm CD$ and $\rm BI$ are valid on $\mathcal M$. However, by Proposition~\ref{PropConstoBI}, $\rm BI$ is a consequence of $\rm CD$, so we only check the latter.
Suppose that $x\in \val {\ubox (\varphi \vee \psi )}$, but $x \not \in \val {\ubox \varphi }$. Then, in view of Proposition~\ref{propBoxKripke}, for some $n \geq 0$, $S^n (x) \not \in \val{\varphi}$. It follows that $S^n (x) \in \val{\psi}$, so that $x \in \val {\diam \psi}$.
\end{proof}

Note that the relational (rather than topological) semantics are used in an essential way in the above proof, since Proposition~\ref{propBoxKripke} is not available in the topological setting, and indeed we will show in Proposition~\ref{PropInd} that these axioms are not topologically valid. But before that, let's turn our attention to the Fischer Servi axioms (see Table \ref{tableLogics}).
Recall that $\sf o$ denotes the class of dynamical systems with a continuous and open map.

\begin{theorem}\label{ThmSoundFS}
$\loghomeo \subseteq {\sf ITL}_{\diam\ubox}^{\sf o}$, i.e., the logic $\loghomeo$ is sound for the class of open dynamical systems.
\end{theorem}

\begin{proof}
Let $\mathcal M = (X,\mathcal T, S, \val \cdot)$ be a dynamical topological model where $S$ is an interior map. We check that axiom ${\rm FS}_\mtnext$ is valid on $\mathcal M$.
Suppose that $x \in \val {\tnext \varphi \to \tnext \psi}$, and let $U = \val {\tnext \varphi \to \tnext \psi}$. Since $S$ is open, $V = S [U]$ is a neighbourhood of $S(x)$. Let $y \in V \cap \val \varphi$, and choose $z\in U$ so that $y = S(z)$. Then, $z\in U\cap \val {\tnext \varphi}$, so that $S(z) = y \in \val {\tnext \psi}$. Since $y\in V$ was arbitrary, $S(x) \in \val{\varphi \to \psi}$, and $x\in \val{\tnext (\varphi \to \psi)}$.
\end{proof}

As an easy consequence, we mention the following combination of Theorems \ref{ThmSoundCD} and \ref{ThmSoundFS}. Recall that dynamic posets with an interior map are also called {\em persistent,} and the class of persistent posets is denoted $\sf p$.

\begin{corollary}\label{CorSoundOne}
The logics $\logpers$ and $\logpersb$ are sound for the class of persistent posets, that is, $\logpers \subseteq \itlp$ and $\logpersb \subseteq \itlpb$.
\end{corollary}

\section{Euclidean spaces}\label{secEuclid}

The celebrated McKinsey-Tarski theorem states that intuitionistic propositional logic is complete for the real line, and more generally for a wide class of metric spaces which includes every Euclidean space $\mathbb R^n$ \cite{tarski}.
Thus, it is natural to ask if a similar result holds for intuitionistic temporal logics, which could lead to applications in spatio-temporal reasoning.
As we will see in this section, the answer to this question is negative; however, we identify some principles which could lead to an axiomatization for Euclidean systems.

Let us consider the real line.
The conditional excluded middle axiom shows that even the $\tnext$-logic of the real line is different from the logic of all dynamical systems.

\begin{lemma}\label{lemmCemRealLine}
The formula
\[{\rm CEM}  (p,q) = ( \neg \tnext p \wedge \tnext \neg \neg p ) \rightarrow ( \tnext q \vee \neg \tnext q )\]
is valid on $\mathbb R$.
\end{lemma}

\proof
Suppose that $(\mathbb R,S,\val\cdot)$ is a model based on $\mathbb R$ and that $x\in \val {\neg \tnext p \wedge \tnext \neg \neg p}$.
From $x\in \val { \tnext \neg \neg p}$ and the semantics of double negation (discussed in \cite{FernandezITLc}) we see that there is a neighbourhood $V$ of $S(x)$ such that $V\subseteq \overline{\val p}$. It follows from the intermediate value theorem that if $U$ is a neighbourhood of $x$ and $S[U]$ is not a singleton, then $S[U] \cap V$ contains an open set and hence $S [U] \cap \val p \not = \varnothing$.
Meanwhile, from $x\in \val{\neg \tnext p}$ we see that $x$ has a neighbourhood $U_\ast$ such that $S[U_\ast] \cap \val p = \varnothing$, hence for such a $U_\ast$ the set $S[U_\ast]$ is the singleton $ \{ S (x) \} $.
But then either $S(x) \in \val q$ and $x \in \val{\tnext q }$, or else $ S (x) \not \in \val q$, which means that $U_\ast \cap  \val{\tnext q } = \varnothing$ and thus $U_\ast$ witnesses that $x\in \val{\neg \tnext q}$.
In either case, $x\in \val{\tnext q \vee \neg \tnext q}$, as required.
\endproof

\begin{remark}
It is tempting to conjecture that ${\rm CEM}$ axiomatizes the $\mathcal L_\tnext$-logic of the real line, but there is a possibility that additional axioms are required.
${\rm CEM} $ is an intuitionistic variant of similar formulas in \cite{KremerMints,SlavnovCounterexamples} showing that the {\em dynamic topological logic} ($\sf DTL$) of the real line is different from the dynamic topological logic of arbitrary spaces.
We will not review $\sf DTL$ here, but it is a classical cousin of intuitionistic temporal logic; see \cite{FernandezITLc,DieguezCompleteness}.
The problem of axiomatizing $\sf DTL$ over the real line has long remained open, and \cite{Nogin} give further examples of valid formulas not derivable from the classical analogue of $\rm CEM$.
We do not know if these formulas also have intuitionistic counterparts, and leave this line of inquiry open.
\end{remark}

Since ${\rm CEM} $ is not valid for $\mathbb R^n$ in general, it cannot be used to show that our base logic is incomplete for Euclidean spaces.
However, this can be shown using ${\rm CD}^-$, which is valid on any {\em locally connected} space.
Recall that a subset $C$ of a topological space $X$ is {\em connected} if, whenever $A,B$ are disjoint open sets such that $C\subseteq A \cup B$, it follows that $C \subseteq A$ or $C \subseteq B$.
The space $X$ is {\em locally connected} if whenever $U$ is open and $x\in U$, there is a connected neighbourhood $V\subseteq U$ of $x$.
It is well-known that $\mathbb R^n$ is locally connected for all $n$.

The following properties regarding connectedness will be useful below.
Suppose that $X$ is a topological space and $S\colon X\to X$ is continuous.
Then, for every $C\subseteq X$, if $C$ is connected, then so is $S[C]$.
Moreover, any $A\subseteq X$ can be partitioned into a family of maximal connected sets called the {\em connected components} of $A$.
If $X$ is locally connected and $A$ is open, then the connected components of $A$ are also open (see e.g.~\cite{Dugundji}).

\begin{lemma}
The formulas $\AxConsm p $ and $t_{\khence} ( \AxConsm p ) $ are valid on the class of locally connected spaces.
\end{lemma}

\proof
Let $(X,\mathcal T,S,\val\cdot)$ be a model based on a locally connected space, and $x\in X$.
First we show that $x\in \val{t_{\khence} ( \AxConsm p )}$.
Recall that $\AxConsm p = \ubox( p\vee\neg  p) \to  \ubox \neg p \vee \diam p $, so that $t_{\khence} ( \AxConsm p ) = {\khence} ( p\vee\neg  p ) \to  {\khence} \neg p \vee \diam p  $.
We may thus assume that $x\in \val{{\khence}(p\vee\neg p)} $, so that using the local connectedness of $X$, there is a connected neighbourhood $U$ of $x$ with $U\subseteq \bigcap_{n\in\mathbb N} S^{-n} \val {p\vee \neg p}$.
This means that, for each $n\in\mathbb N$, $U\subseteq S^{-n} \val {p\vee \neg p} = S^{-n} \val {p}\cup S^{-n} \val { \neg p}$ and, from the connectedness of $U$, $U\subseteq   S^{-n} \val {p} $ or $U\subseteq S^{-n} \val { \neg p} $, as these two sets are disjoint and open.
If $x\in \val{\diam p}$ there is nothing to prove, so we assume otherwise.
This means that for all $n$, $ S^n(x) \not \in \val p$, which implies that $U \not \subseteq   S^{-n} \val {p} $, and hence $U\subseteq S^{-n} \val { \neg p} $.
But then, $U$ witnesses that $x\in \val{{\khence} \neg p}$.

To see that $x\in \val{  \AxConsm p  }$, suppose that $x\in \val{\ubox(p\vee\neg  p)} $; we must show that $x\in \val {\ubox \neg p \vee \diam p }$.
Let $U \subseteq \val{p\vee \neg p}$ be an $S$-invariant neighbourhood of $x$.
For $n \in \mathbb N$, let $V_n$ be the connected component of $U$ containing $S^n(x)$, and set $V = \bigcup_{n\in \mathbb N} V_n$.
Since each $V_n\subseteq U \subseteq \val p \cup \val {\neg p}$ and the latter are disjoint and open, it follows that either $V_n\subseteq \val p$ or $V_n\subseteq \val {\neg p}$.
If $V_n\subseteq \val p$ for some $n$, it immediately follows that $x\in \val{\diam p}$, and we are done.
Otherwise, $V_n\subseteq \val {\neg p}$ for all $n$.
Clearly $V$ is open; we claim that it is also $S$-invariant.
To see this, note that $S[V_n] \subseteq U$.
Note that $U$ can be written as the disjoint union of two open sets as $U = V_{n+1} \cup(U\setminus V_{n+1}) $; since $V_n$ is connected, so is $S[V_n]$, hence $S[V_n] \subseteq V_{n+1}$ or $S[V_n] \subseteq U\setminus V_{n+1} $.
However, $S[V_n] \cap V_{n+1}$ is non-empty, so we must have $S[V_n] \subseteq V_{n+1}$, and since $n$ was arbitrary, $S[V] \subseteq V$, as claimed.
Hence $V$ witnesses that $x\in \val{\ubox \neg p}$, as needed.\endproof

In conclusion, we obtain the following.

\begin{theorem}\
\begin{enumerate}

\item $\logrealax$ and $\logrealaxa$ are sound for $\mathbb R$.

\item $\logrnax$ and $\logrnaxa$ are sound for $\{\mathbb R^n:n > 0\}$.

\item $\logrnfs$ and $\logrnfsa$ are sound for the class of invertible systems based on $\{\mathbb R^n:n > 0\}$.

\end{enumerate}
\end{theorem}

In the remainder of this section we show that ${\sf ITL}_{\diam\ubox{\khence}}^{\mathbb R^2} \subseteq {\sf ITL}^{\sf e}_{\diam\ubox{\khence}} \cap {\sf ITL}_{\diam\ubox{\khence}}^{\mathbb R }$. 
We show this using results from \cite{FernandezR2} originally developed for dynamic topological logic, but applicable to $\sf ITL$ as well.
We begin with the notion of dynamic morphism.

\begin{definition}
Let $\mathcal X = (X,\mathcal T_\mathcal X,S_\mathcal X)$ and $\mathcal Y = (Y ,\mathcal T_\mathcal Y,S_\mathcal Y)$ be dynamic topological systems.
Let $U\subseteq X$ be open and $S_\mathcal X$-invariant.
A {\em dynamic morphism} from $\mathcal X$ to $\mathcal Y$ is an interior map $\pimor \colon U\to Y$ such that for all $x\in X$, $\pimor S_\mathcal X(x) = S_\mathcal Y \pimor(x)$.
\end{definition}

\begin{proposition}\label{propPMorph}
Let $\mathcal X = (X,\mathcal T_\mathcal X,S_\mathcal X)$ and $\mathcal Y = (Y ,\mathcal T_\mathcal Y,S_\mathcal Y)$ be dynamic topological systems.
Let $U\subset X$ be open and $S_\mathcal X$-invariant and $\pimor \colon U \to Y$ be a dynamic morphism, and let $\val\cdot_\mathcal Y$ be any valuation on $\mathcal Y$.
Then, there is a valuation $\val\cdot_\mathcal X  $ on $\mathcal X$ such that for every formula $\varphi \in \mathcal L_{\diam\ubox{\khence}}$, $\val \varphi_\mathcal X \cap U = \pimor^{-1} \val \varphi_\mathcal Y$.
\end{proposition}

\begin{proof}
Let $\val \cdot_\mathcal X$ be the unique valuation such that for any propositional variable $p$, $\val p_\mathcal X = \pimor^{-1} \val p_\mathcal Y$.
We prove by induction on $\varphi$ that $\val \varphi_\mathcal X \cap U = \pimor^{-1} \val \varphi_\mathcal Y$.
Most cases are standard, save the cases for $\varphi = {\khence} \psi$ and $\varphi = \ubox \psi$, so we focus on those.
First assume that $x\in U \cap \val {{\khence} \psi}_\mathcal X$.
Let $V\subseteq U$ be a neighbourhood of $x$ such that $V\subseteq \bigcap_{n\geq 0} S_\mathcal X^{-n} \val \psi_\mathcal X$.
Then $\pimor[V]$ is open, and since $\pimor S_\mathcal X = S_\mathcal Y \pimor$,
\[\pimor[V] \subseteq  \pimor \Big [ \bigcap_{n\geq 0} S_\mathcal X^{-n} \val \psi_\mathcal X \Big ] \subseteq \bigcap_{n\geq 0} S_\mathcal X ^{-n} \pimor \val \psi_\mathcal X \stackrel{\text{\sc ih}}{\subseteq} \bigcap_{n\geq 0} S_\mathcal Y ^{-n}  \val \psi_\mathcal Y .\]
So, $\pimor[V]$ witnesses that $\pimor(x) \in \val \psi_\mathcal Y$.

If $\pimor(x) \in \val {{\khence} \psi}_\mathcal Y$, we instead let $V'$ be a neighbourhood of $\pimor(x)$ so that $V'\subseteq \bigcap _{n\geq 0} S_\mathcal Y^{-n} \val \psi_\mathcal Y$.
Then,
\[ \pimor^{-1} [ V' ] \subseteq \pimor^{-1} \Big [ \bigcap _{n\geq 0} S_\mathcal Y ^{-n} \val \psi_\mathcal Y \Big ] = \bigcap _{n\geq 0} S_\mathcal Y^{-n} \pimor ^{-1} \val \psi_\mathcal Y \stackrel{\text{\sc ih}}{\subseteq}\bigcap _{n\geq 0} S_\mathcal X^{-n}  \val \psi_\mathcal X .\]
Since $\pimor^{-1} [ V' ]$ is open, this witnesses that $x\in \val {{\khence} \psi}$.

The arguments for $\ubox \psi$ are similar.
As above, if $x\in \val {\ubox \psi}_\mathcal X$, we let $V\subseteq \val \psi_\mathcal X$ be an open, $S_\mathcal X$-invariant neighbourhood of $x$.
Since $U$ is open and $S_\mathcal X$-invariant, $V\cap U$ is also open and $S_\mathcal X$-invariant, so we may assume that $V\subseteq U$.
Then, $\pimor[V]$ is open and $\pimor[V] \subseteq \val \psi_\mathcal Y$ by the induction hypothesis.
It remains to check that $\pimor[V]$ is $S_\mathcal Y$-invariant.
But if $y\in \pimor[V]$ then $y = \pimor(z)$ for some $z\in V$, hence $S_\mathcal Y(y) = S_\mathcal Y\pimor(z) =\pimor S_\mathcal X (z)  $ and $S_\mathcal X (z) \in V$ by $S_\mathcal X$-invariance, so $S_\mathcal Y(y) \in \pimor[V]$, as needed.

Similarly, if $V'$ witnesses that $\pimor(x) \in \val {\ubox \psi}_\mathcal Y$, then $\pimor^{-1}[V']$ witnesses that $x\in \val{\ubox \psi}_\mathcal X$.
\end{proof}

From here, we easily obtain that the logics based on $\mathbb R^2$ are contained in those based on $\mathbb R$:

\begin{theorem}\label{theoR2R}
Every formula of $\mathcal L_{\diam\ubox\khence}$ valid on $\mathbb R^2$ is valid on $\mathbb R$; that is, ${\sf ITL}^{\mathbb R^2}_{\diam\ubox{\khence}} \subseteq {\sf ITL}^{\mathbb R }_{\diam\ubox{\khence}} $.
\end{theorem}

\begin{proof}
Suppose that $\varphi $ is not valid on $\mathbb R$.
Let $S \colon\mathbb R\to \mathbb R$ and $\val\cdot$ be such that $(\mathbb R,S,\val\cdot)\not\models\varphi$.
Define $\pimor \colon \mathbb R^2 \to \mathbb R^2$ to be given by $\pimor(x,y) = x$ and $ S' (x,y) = (S(x),y) $.
Then, it is not hard to see that $\pimor$ is a surjective dynamic morphism from $(\mathbb R ^2,S') $ onto $(\mathbb R, S)$.
It follows that $(\mathbb R^2,S',\pimor^{-1}\val\cdot) \not \models \varphi$, and hence $\varphi \not \in {\sf ITL}^{\mathbb R^2}_{\diam\ubox{\khence}}$.
\end{proof}

In order to show that $ {\sf ITL}_{\diam\ubox\khence}^{\mathbb R^2} \subseteq {\sf ITL}^{\sf e}_{\diam\ubox\khence}$, it would suffice to construct a dynamic morphism from $\mathbb R^2$ onto a given dynamic poset $(W,{\peq},S)$.
We may assume that $W$ is finite in view of the following result, which is established by \cite{BalbianiToCL}.

\begin{theorem}\label{theoFMP}
Any formula satisfiable (falsifiable) on a dynamic poset is satisfiable (falsifiable) on a finite dynamic poset.
\end{theorem}

We may then use the following result, essentially proven in \cite{FernandezR2}.

\begin{theorem}\label{theoIsPMorph}
If $(W,{\peq},S)$ is a finite dynamic poset and $w_0 \in W$, there exist continuous $ T\colon \mathbb R^2 \to \mathbb R^2$, an open, $T$-invariant set $U\subseteq \mathbb R^2$, and a dynamic morphism $\pimor\colon \mathbb R^2 \to W$, such that $0\in U$ and $\pimor(0) = w_0$.
\end{theorem}

\begin{proof}[Proof sketch.]
\cite{FernandezR2} shows the result for any {\em dynamic preorder with limits that commute with $S$.}
Recall that $\peq$ is a preorder if it is transitive and reflexive (but not necessarily antisymmetric).
We say that $(W,{\peq},S)$ is a dynamic preorder if $S \colon W\to W$ is $\peq$-monotone, so that dynamic posets are a special case of dynamic preorders.
We say that $(W,{\peq},S)$ is a {\em dynamic preorder with limits} if every {\em monotone} sequence $(w_i)_{i\in \mathbb N}$ (i.e., a sequence such that $w_i\peq w_{i+1}$ for all $i$) is assigned a {\em limit} $\lim_{i\to \infty} w_i \in W$ with the property that $w_n \peq \lim_{i\to \infty} w_i $ for all $n$ and $\lim_{i\to \infty} w_i \peq w_n$ for $n$ large enough.
The limits {\em commute with $S$  if $S\left ( \lim_{i\to \infty} w_i \right) = \lim_{i\to \infty} S( w_i )$.}

However, if $W$ is finite and $\peq$ is a partial order then every monotone sequence can trivially be assigned a limit, as in this case we have that $w_i$ is constant for $i$ large enough, and we can define $\lim_{i\to \infty} w_i$ to be this unique constant.
It is easy to see that this is a limit assignment that commutes with $S$, hence the desired $U$, $T$ and $\pimor$ exist by \cite{FernandezR2}.
\end{proof}

\begin{theorem}\label{theoR2vsITLe}
Every formula of $\mathcal L_{\diam\ubox\khence}$ valid on $\mathbb R^2$ is valid on the class of expanding posets; that is, ${\sf ITL}_{\diam\ubox{\khence}}^{\mathbb R^2} \subseteq {\sf ITL}^{\sf e}_{\diam\ubox{\khence}}$.
\end{theorem}

\proof
We prove the claim by contrapositive.
If $\varphi\not \in {\sf ITL}^{\sf e}_{\diam\ubox{\khence}}$, then by Theorem \ref{theoFMP} there are a finite dynamic poset model $\mathcal M = (W,{\peq},T,\val\cdot_W)$ and $w_0 \in W \setminus \val \varphi_W$.
By Theorem \ref{theoIsPMorph}, there exist $ T\colon \mathbb R^2 \to \mathbb R^2$, an open, $T$-invariant set $U\subseteq \mathbb R^2$, and a dynamic morphism $\pimor\colon U \to W$, such that $0\in U$ and $\pimor(0) = w_0$.
By Proposition \ref{propPMorph}, there is a valuation $\val\cdot _{\mathbb R^2}$ on $\mathbb R^2$ such that $\val \varphi _{\mathbb R^2} \cap U = \pimor^{-1} \val \varphi_W$; in particular, since $w_0 \not \in \val \varphi_W$, we have that $0\not \in \val \varphi _{\mathbb R^2}$, so that $\mathbb R^2\not \models \varphi$ and hence $\varphi \not \in {\sf ITL}_{\diam\ubox{\khence}}^{\mathbb R^2} $.
\endproof

\section{Persistent posets and the real line}\label{secPers}

We have seen that ${\sf ITL}_{\diam\ubox\khence}^{\mathbb R^2} \subseteq {\sf ITL}_{\diam\ubox\khence}^{\sf e}$.
The question naturally arises whether a similar result holds when restricting to open systems: is every formula falsifiable on a persistent poset (i.e., a poset equipped with a continuous, open map) falsifiable on some Euclidean space, also equipped with a continuous, open map?
Surprisingly, not only is the answer affirmative, but in this case, any falsifiable formula is falsifiable on the real line.
In this section, we will prove this fact, along with some properties of $\itlp$ which may be interesting on their own right.
We remark that in this context $\ubox$ and $\khence$ coincide, so we restrict our attention to languages with the former.

One challenge is that we do not have the finite model property in this setting.
This is already proven in \cite{BalbianiToCL} for $\mathcal L_{\diam\ubox}$, but in fact, the finite model property already fails over $\mathcal L_\ubox$.

\begin{proposition}\label{propNoFMP}
The formula $\varphi = \ubox \neg\neg p \to \neg \neg \ubox p$ is valid over the class of all finite persistent posets, but not over the class of all persistent posets.
\end{proposition}

\begin{proof}
First we show that $\varphi$ is indeed valid over any finite persistent poset.
Let $ \mathcal M = (W,\peq,S,\val\cdot)$ be any model based on a finite persistent poset, and let $w\in \val{\ubox\neg\neg p}$.
We show that $w\in \val{\neg\neg\ubox p}$.
It suffices to show that if $v\seq w$ is maximal, then $v \in \val{\ubox p}$.\footnote{In fact, the only property we use of $\mathcal M$ is that for every $w$ there is a maximal $v\seq w$, so $\varphi$ is valid over any persistent model with this property.}
So, let $n\in \mathbb N$.
Then, $S^n(w) \in \val {\neg\neg p}$, and $S^n(v)$ is maximal (as order-preserving persistent functions preserve maximality), so $S^n(v) \in \val p$.
It follows that $v\in \val{\ubox p}$.

To see that $\varphi$ is not valid over the class of persistent posets, let $W =  \mathbb Z$, where $\peq$ is the usual order and $S(x) = x-1$.
Let $\val p =\mathbb N$.
Then, every point satisfies $\neg\neg p$ (as every large-enough point satisfies $p$), so that in particular $ 0 \in \val {\ubox \neg \neg p} $.
However, no point satisfies $\ubox p$, since $S^{x+1}(x) \not \in \val p$.
Hence, in particular, $0 \not \in \val \varphi$.
\end{proof}

Thus, we cannot avoid working with infinite models.
However, we can still work with models that have some `nice' properties.
For starters, $\itlp$ is complete for the class of product models \cite{KuruczMany}.
The products we consider will have a rather particular form.
We will need the following general definition, which applies to arbitrary topological spaces.

\begin{definition}\label{defInfty}
Let $\mathcal X = (X, \mathcal T )$ be any topological space.
We define a new dynamical system $ \timesn{ \mathcal X } = (\timesn X, \timesn{\mathcal T} ,\timesn S)$, where
\begin{itemize}

\item $\timesn X = X  \times\mathbb N$,

\item  $U\subseteq  \timesn X$ is open if and only if for every $n\in \mathbb N$, $\{x\in X: (x,n) \in U\}$ is open, and

\item  $\timesn S(x,n) = (x,n+1)$.

\end{itemize}
We say that a dynamical system $\mathcal Y$ is a {\em product system} if $\mathcal Y =\timesn{\mathcal X }$ for some space $\mathcal X$.
If $\mathcal X$ was a poset, then $\mathcal Y$ is a {\em product poset.}
\end{definition}

Products have the property that they `lift' interior maps to dynamic morphisms, in the following sense:

\begin{lemma}\label{lemmProduct}
Let $\mathcal X = ( X,\mathcal T _\mathcal X )$ be a topological space, $\mathcal Y = ( Y,\mathcal T_\mathcal Y,S_\mathcal Y)$ be an open dynamical system, and suppose that $\pimor\colon X \to Y$ is an interior map.
Then, there exists a dynamic morphism $\timesn\pimor  \colon \timesn X \to Y$ whose range is $\bigcup_{n\in\mathbb N} S^n_\mathcal Y \pimor[X]$.
\end{lemma}

\begin{proof}
Define $\timesn \pimor  \colon X\times \mathbb N \to Y$ by $\timesn \pimor (x,n) = S^n_\mathcal Y \pimor (x)$.
First we must check that $\timesn \pimor$ is continuous and open.
If $ (x,n) \in X_\infty $ and $U$ is a neighbourhood of $\timesn \pimor (x,n)$, then since both $\pimor$ and $S^{n}_\mathcal Y$ are continuous, $\pimor^{-1} S^{-n}_\mathcal Y [U] \times \{n\}$ is a neighbourhood of $(x,n)$ contained in $\timesn{\pimor^{-1}}[U]$.
Since $(x,n)$ was arbitrary, $ \timesn\pimor$ is continuous.
Similarly, if $O$ is a neighbourhood of $(x,n)$, then $S^n_\mathcal Y \pimor [O\cap(X\times \{n\})]$ is a neighbourhood of $\timesn\pimor(x,n)$ contained in $\timesn \pimor[O]$, which since $(x,n)$ was arbitrary shows that $\timesn \pimor[O]$ is open.
Hence, $\timesn\pimor$ is an open map.

Next we note $S_\mathcal Y\timesn \pimor(x,n) = S_\mathcal Y \circ S_\mathcal Y^{n} \pimor(x) =S_\mathcal Y^{n+1} \pimor(x) = \timesn\pimor (x,n+1) = \timesn\pimor \timesn S(x,n)$.
Finally, the range of $\timesn \pimor$ is
\[\timesn \pimor[X\times \mathbb N] = \bigcup_{n\in\mathbb N}\timesn \pimor [X\times \{n\}] = \bigcup_{n\in\mathbb N} S^n_\mathcal Y \pimor[X].\]
\end{proof}

As a corollary, we immediately obtain that any formula satisfiable (falsifiable) on a persistent poset is satisfiable (falsifiable) on a product poset, since we can take $\mathcal X = \mathcal Y $ and $\pimor$ to be the identity.
With this, we can easily check that we can restrict our attention to the class of countable models.
Below, if $\mathcal M = (W,\peq,S,\val\cdot)$ is a dynamical poset model and $Z\subseteq W$, then $\mathcal M \upharpoonright Z = (Z,\peq\upharpoonright Z,S \upharpoonright Z ,\val\cdot \upharpoonright Z)$ is the sub-structure with domain $W$ and such that each of $\peq$, $S$ and $\val\cdot$ are restricted to $Z$.
If $Z$ is $S$-invariant but not necessarily open, then $\mathcal M \upharpoonright Z$ is also based on a dynamic poset, although $\val\cdot \upharpoonright Z$ may not be a valuation.
However, this will indeed be the case if we choose $Z$ appropriately.

\begin{lemma}\label{lemLowen}
If $\varphi \in \mathcal L_{\diam\ubox}$ is falsifiable on a persistent poset, it is falsifiable on a countable model.
\end{lemma}

\begin{proof}
We proceed as in a standard proof of the downward L\"owenheim-Skolem theorem.
Let $\mathcal M = (W,\peq,S,\val\cdot)$ be a model based on a persistent poset such that $\mathcal M \not\models \varphi$.
In view of Lemma \ref{lemmProduct}, we may assume that $\mathcal M$ is a product model, so that $W = \timesn U $ for some $U$, and there is $w_0 \in U$ such that $(w_0,0)\not\in \val\varphi$.
We define a sequence $V_0\subseteq V_1\subseteq \ldots \subseteq U$ of countable sets, so that for $V = \bigcup_{n\in \mathbb N} V_n$ we have that $\mathcal M\upharpoonright V$ falsifies $\varphi$.
The construction is straightforward: $V_0 = \{w_0\}\times\mathbb N$, and if we are given $V_n$, define
\[V_{n+1} := V_n \cup \{(v _{w,k}^{\psi\to\theta},m) :\text{ $ (w,k) \in V_n$, $ \psi, \theta \in \mathcal L_{\diam\ubox} $, and $ m\in\mathbb N$}\},\]
where $v _{w,k}^{\psi\to\theta} = w$ if $ (w,k) \in \val {\psi\to \theta}$, and otherwise $v = v _{w,k}^{\psi\to\theta} $ is chosen to satisfy $ v  \seq w$, $ (v,k) \in \val \psi $ and $(v,k) \not \in \val \theta$.
Note that at each stage we add $(v _{w,k}^{\psi\to\theta} ,m)$ for {\em all} $m$, which ensures that the resulting set is $S$-invariant and that $S$ is open on each $V_n$, hence on $V$.
One can then easily check that $w_0$ satisfies the same formulas on $\mathcal M \upharpoonright V$ as it did on $\mathcal M$.
\end{proof}

Given the lack of the finite model property for persistent posets we need to consider infinite posets, but in view of Lemma \ref{lemLowen}, it suffices to work with countable posets.
Fortunately, we may work with a single, `universal' countable poset.

\begin{definition}
Define $2^{<\mathbb N}$ to be the set of (possibly empty) binary strings, where $a \sqsubseteq b$ if and only if $a$ is an initial segment of $b$.
We denote the empty string by $\epsilon$.
\end{definition}

The following is proven in \cite{GoldblattDiodorean} for finite structures, and as mentioned in \cite{KremerStrong}, readily extends to countable structures.

\begin{theorem}\label{theoStringstoW}
Given a countable poset $(W,\peq)$ and $w_0 \in W$, there is an interior map $\pimor \colon 2^{<\mathbb N} \to W$ such that $w_0 = \pimor(\epsilon)$.
\end{theorem}

In view of Lemma \ref{lemmProduct}, one then immediately obtains the following.

\begin{corollary}\label{corStringstoW}
Given a countable persistent poset $(W,\peq,S)$ and $w_0 \in W$, there is a dynamic morphism $\pimor  \colon \timesn{ 2^{<\mathbb N} } \to W$ such that $\pimor(\epsilon) = w_0$.
\end{corollary}

We extend the notation $\sqsubseteq $ to elements of $\timesn{ 2^{<\mathbb N} } = 2^{<\mathbb N} \times \mathbb N$ by letting $(a,n)\sqsubseteq (b,m)$ if $a\sqsubseteq b$ and $n=m$.
It is not hard to check that the order $\sqsubseteq$ generates the topology on $\timesn{ 2^{<\mathbb N} }$.
It would remain to show that there is a dynamic morphism $\pimor \colon U \to \timesn{2^{<\mathbb N} }$ for some suitable $U\subseteq \mathbb R$ and some suitable $S\colon \mathbb R\to \mathbb R$.
Actually, as already observed in \cite{KremerStrong}, this will not be possible.
So, we must replace $  \timesn{2^{<\mathbb N} }$ by a slightly different space.

\begin{definition}

Let $2^{\leq \mathbb N}$ be the set of all finite or infinite strings, and extend the notation $\sqsubseteq$ to $2^{\leq \mathbb N}$ by also setting $a\sqsubseteq b$ if $a$ is an initial segment of $b$.
For finite $b$, set ${\uparrow} b = \{a\in 2^{\leq \mathbb N} : b\sqsubseteq a\}$.
Say that $U\subseteq 2^{\leq \mathbb N}$ is {\em open} if whenever $a \in U$, there is a finite $b\sqsubseteq a$ such that ${\uparrow} b \subseteq U$.
\end{definition}

The space $2^{\leq \mathbb N}$ is thus endowed with a topology.
This space is more convenient since, indeed, there is a surjective interior map from $(0,1)$ to this space.
The following is shown in \cite{KremerStrong}.

\begin{theorem}\label{theoPInt}
There exists a surjective interior map $\pimor\colon (0,1) \to 2^{\leq \mathbb N}$.
\end{theorem}

This immediately gives us a dynamic morphism from an open subset of $\mathbb R$ to the dynamical system $\timesn{2^{\leq \mathbb N} } $ defined as above.
As we did for $\timesn{ 2^{<\mathbb N} }  $, we further extend the notation $\sqsubseteq $ to elements of $\timesn{ 2^{\leq \mathbb N} }  $ by letting $(a,n)\sqsubseteq (b,m)$ if $n=m$ and $a\sqsubseteq b$.

\begin{lemma}\label{lemmRtoStrings}
There exist a function $S\colon \mathbb R\to \mathbb R$, an open, $S$-invariant $U\subseteq \mathbb R$, and a surjective dynamic morphism $\pimor \colon U \to \timesn{ 2^{\leq \mathbb N} }$.
\end{lemma}

\begin{proof}
Let $S$ denote the map on $\timesn {2^{\leq \mathbb N}} $ as given in Definition \ref{defInfty}.
By Theorem \ref{theoPInt}, there is a surjective interior map $\pimor'\colon (0,1)\to 2^{\leq \mathbb N}$, which can be viewed as an interior map $\pimor''\colon (0,1)\to 2^{\leq \mathbb N}\times \mathbb N$ with range $2^{\leq \mathbb N} \times \{0\} $.
Hence Lemma \ref{lemmProduct} tells us that there is a dynamic morphism $\pimor\colon \timesn{(0,1)} \to \timesn{2^{\leq \mathbb N}}$ with range $\bigcup_{n\in\mathbb N} S^n[ 2^{\leq \mathbb N} \times \{0\} ] = 2^{\leq \mathbb N} \times \mathbb N$.
Define $S\colon \mathbb R \to \mathbb R$ by $S(x) = x+1$, and let $U =(0,\infty)\setminus \mathbb N$ (i.e., the set of positive reals that are not integers). 
Then it is easily seen that $ \iota\colon \timesn{ (0,1) } \to U$ given by $(x,n) \mapsto x+n$ is an isomorphism.
Thus $\pimor\iota^{-1} \colon U \to \timesn{2^{\leq \mathbb N} } $ is also a surjective dynamic morphism.
\end{proof}

\begin{remark}
In fact \cite{KremerStrong} proves Theorem \ref{theoPInt} for any complete metric space without isolated points, and shows how the completeness assumption can be dropped using algebraic semantics.
Using Kremer's result, we could generalize Lemma \ref{lemmRtoStrings} to many more spaces, including the rational numbers and the Cantor space.
We require only the existence of suitable $U$ and $S$, so as to be able to adapt the proof.
\end{remark}

Even though the dynamical systems $\timesn{2^{< \mathbb N} }$ and $\timesn{ 2^{\leq \mathbb N} } $ are not isomorphic (the former is countable and the latter is not), they have essentially the same open sets.
This observation will be key in proving the following lemma.

\begin{lemma}\label{lemmPreserve}
A formula of $ \mathcal L_{\diam\ubox}$ is satisfiable (falsifiable) on $\timesn{2^{< \mathbb N} }$ if and only if it is satisfiable (falsifiable) on $\timesn{ 2^{\leq \mathbb N} } $.
\end{lemma}

\begin{proof}
We follow an algebraic approach.
Let $\timesn{2^{< \mathbb N} } = (2^{< \mathbb N} \times\mathbb N,\mathcal T_<,S_<)$ and $\timesn{2^{\leq  \mathbb N} } = (2^{ \leq \mathbb N} \times\mathbb N ,\mathcal T_\leq ,S_\leq )$.
We show that $\mathcal T_<$ and $\mathcal T_\leq$ are isomorophic as Heyting algebras by providing order-preserving maps $ f\colon \mathcal T_< \to \mathcal T_\leq$ and $ g\colon \mathcal T_\leq \to \mathcal T_<$ so that $g$ is the inverse of $f$.
Moreover, we show that $f,g$ commute with preimages under $S_\cdot$ and preserve $S_\cdot$-invariant open sets, where $\cdot \in \{<,\leq\}$.
From this it will readily follow that if $\val\cdot_<$ is a valuation on $\timesn{2^{< \mathbb N} }$ then $f\circ \val\cdot_<$ is a valuation on $\timesn{2^{\leq \mathbb N} }$, and similarly if $\val\cdot_\leq $ is a valuation on $\timesn{2^{\leq \mathbb N} }$ then $g\circ \val\cdot_\leq $ is a valuation on $\timesn{2^{<\mathbb N} }$, and hence the two structures have the same valid formulas.

So, for $U \in \mathcal T_<$ define $f(U) = \{a\in X_\leq: \exists b\sqsubseteq a \ b\in U\}$, and for $V \in \mathcal T_\leq$ define $g(V)$ to be the set of finite elements of $V$.
We need the following properties of $f,g$.
\medskip

\noindent{\sc $g\circ f$ is the identity on $\mathcal T_<$:} If $a\in U \in \mathcal T_<$ then from $a\sqsubseteq a$ we obtain $a\in f(U)$, and since $a$ is finite, $a\in g\circ f(U)$.
Conversely, if $a\in g\circ f(U)$ then $a$ is finite and $a\in f(U)$, hence $\exists b\sqsubseteq a$ such that $b\in U$, but $U$ is upward-closed, so $a\in U$.
\medskip

\noindent{\sc $f\circ g$ is the identity on $\mathcal T_\leq$:} if $a\in V  \in \mathcal T_\leq$ then since $V$ is open there is finite $b\sqsubseteq a$ such that $b\in V$, so that $b\in g(V)$; but $b\sqsubseteq a$ implies that $a\in f\circ g(V)$.
Conversely, if $a\in f\circ g(V)$ then there is $b\sqsubseteq a$ which is finite and so that $b\in g(V)$, which implies $b\in V$ and, since $V$ is upwards-closed, $a\in V$.\medskip

\noindent{\sc $f$ is order-preserving:} If $U\subseteq U'\in \mathcal T_<$ and $a\in f(U)$ there is finite $b\in U$ such that $b\sqsubseteq a$, but then $b\in U'$ and $a\in f(U')$.
\medskip

\noindent{\sc $g$ is order-preserving:} If $V\subseteq V' \in \mathcal T_\leq$, then clearly every finite element of $V$ is a finite element of $V'$, so $g(V) \subseteq g(V')$.\medskip

\noindent{\sc $f\circ S^{-1}_< = S^{-1}_\leq \circ f $:}
Let $U\in \mathcal  T_<$ and $(a,n) \in 2^{< N} \times \mathbb N$.
If $(a,n) \in f\circ S^{-1}_<[U]$, there is finite $b\sqsubseteq a$ such that $(b,n)\in S^{-1}_<[U]$, so that $S_<(b,n) =(b,n+1) \in U$, witnessing that $S_\leq(a,n) = (a, n+1) \in f(U)$, hence $(a,n)\in S^{-1}_\leq \circ f(U)$.
Conversely, if $(a,n) \in S^{-1}_\leq \circ f (U)$, then $(a, n+1) \in f(U)$, and so there is finite $c\sqsubseteq a$ with 
$(c, n+1) \in U$, so that $(c,n) \in S^{-1}_\leq [U]$ and $(a,n) \in f\circ S^{-1}_\leq [U]$.\medskip

\noindent{\sc $g\circ S^{-1}_\leq = S^{-1}_< \circ g $:} Take $V\in \mathcal T_\leq$ and $ a  \in 2^{\leq N} \times \mathbb N$.
If $ a \in g\circ S^{-1}_\leq [V]$, then $a$ is finite and $S_\leq (a ) \in V$; but $S_\leq (a )$ is also finite so $S_< (a ) = S_\leq (a ) \in g(V)$, hence $a\in S^{-1}_< \circ g(V)$.
Conversely, if $a\in S^{-1}_< \circ g(V)$, then $S_<(a) \in g(V)$, so that $S_<(a) \in V$ and $a\in S^{-1}_<[V]$, which since $a$ is finite implies $a\in g\circ S^{-1}_<[V]$.
\medskip

\noindent{\sc $f$ preserves invariant sets:}
If $U \in \mathcal T_<$ is $S_<$-invariant and $ (a,n)  \in f(U)$ then there is finite $b\sqsubseteq a$ so that $(b,n)\in U$, hence $(b,n+1) \in U$, witnessing that $S_\leq (a,n) = (a,n+1) \in f(U)$, and $f(U)$ is $S_\leq$-invariant.
\medskip

\noindent{\sc $g$ preserves invariant sets:} If $V \in\mathcal T_\leq$ is $S_\leq$-invariant and $a \in g(U)$, then $a$ is finite and $S_<(a) = S_\leq (a) \in V$, which since $a$ is finite implies that $S_\leq (a)$ is finite so that $S_\leq (a) \in g(V)$.
\medskip

If $\val\cdot_<$ is a valuation on $\timesn{2^{<\mathbb N}}$, then the above considerations show that $f\circ\val\cdot_< $ is a valuation on $\timesn{2^{\leq \mathbb N}}$, and conversely, if $\val\cdot_\leq $ is a valuation on $\timesn{2^{\leq \mathbb N}}$, then $g\circ\val\cdot_\leq  $ is a valuation on $\timesn{2^{<\mathbb N}}$.
For example, we show that $f\circ\val\cdot_< $ satisfies the clause for $\ubox$.
We have that $\val{\ubox\varphi}_< $ is $S_<$-invariant, so that $ f ( \val{\ubox\varphi}_< ) $ is $S_\leq$-invariant; that $  \val{\ubox\varphi}_< \subseteq \val{ \varphi}_<$, so that $ f ( \val{\ubox\varphi}_< ) \subseteq f ( \val{ \varphi}_< ) $; and finally, if $U\in \mathcal T_\leq$ is $S_\leq$-invariant and $U\subseteq  f ( \val{ \varphi}_<)$, then $g(U)$ is $S_<$-invariant and $g(U) \subseteq \val{ \varphi}_<$, so that $g(U) \subseteq  \val{ \ubox \varphi}_<$, hence $U = f\circ g(U) \subseteq f ( \val{ \ubox \varphi}_<)$.
We conclude that $f ( \val{\ubox\varphi}_<) $ is the greatest $S_\leq$-invariant open set contained in $f ( \val\varphi_<)$, as needed.

Finally, we observe that if $ \val\varphi_< \neq 2^{<\mathbb N }\times \mathbb N$, then also $ f ( \val\varphi_< ) \neq 2^{\leq\mathbb N }\times \mathbb N$, since $f$ is injective and $f(2^{<\mathbb N }\times \mathbb N) = 2^{\leq \mathbb N }\times \mathbb N$, so any formula falsifiable on $\timesn{ 2^{<\mathbb N }}$ is falsifiable on $\timesn{ 2^{<\mathbb N }}$; and, similarly, if $ \val\varphi_\leq  \neq 2^{\leq \mathbb N }\times \mathbb N$ then also $ g ( \val\varphi_\leq ) \neq 2^{ < \mathbb N }\times \mathbb N$.
We conclude that $\varphi $ is valid on $\timesn{ 2^{<\mathbb N }}$ if and only if it is valid on $\timesn{ 2^{\leq\mathbb N }}$.
\end{proof}

Putting our results together, we obtain the following.

\begin{theorem}\label{theoRcomplete}
Every formula of $\mathcal L_{\diam\ubox}$ valid over $\mathbb R$ equipped with an interior map is valid over the class of persistent posets; that is, ${\sf ITL}^{{\sf o}\cap \mathbb R}_{\diam\ubox} \subseteq \itlp$.
\end{theorem}

\begin{proof}
If $\varphi$ is falsifiable on some persitent poset, then by Corollary \ref{corStringstoW}, it is falsifiable on $2^{<\mathbb N}$, hence by Lemma \ref{lemmPreserve}, on $2^{\leq \mathbb N}$.
From this and Lemma \ref{lemmRtoStrings}, it follows that $\varphi$ is falsifiable on $\mathbb R$ with a homeomorphism.
\end{proof}

\begin{remark}
Compare the situation to that of $\sf DTL$, where \cite{FernandezMetric} showed that there is a formula that is sound for the class of complete metric spaces with an open map, but not sound for the class of persistent posets.
The latter is due to the Baire category theorem; in view of Theorem \ref{theoRcomplete}, the Baire category theorem does not seem to affect the intuitionistic temporal logic of $\mathbb R$.
\end{remark}

Figure \ref{fig:sem} summarizes the inclusions between the semantically-defined logics we have defined. As we will see in the remainder of this paper, these are the only inclusions that hold between said logics.

\begin{figure}[H]\centering
	\begin{tikzpicture}[thick,->,auto,font=\small,node distance=1.5cm]
	\node[] (logbasic) {${\sf ITL}^{\sf c}_{\diam\ubox\khence}$};
	\node[] (logrnax) 	[above left of=logbasic] {${\sf ITL}^{ \mathbb R^2}_{\diam\ubox\khence}$};	
	\node[] (logrealax) [above right of=logrnax] {${\sf ITL}^\mathbb R_{\diam\ubox\khence}$};	
	\node[] (logexp) [above left of=logrnax] {${\sf ITL}^{\sf e}_{\diam\ubox\khence}$};
	\node[] (logrnfs)  [above right of=logrealax] {${\sf ITL}^{{\sf o}\cap\mathbb R}_{\diam\ubox\khence}$};		
	\node[] (loghomeo)  [below right of=logrnfs] {${\sf ITL}^{\sf o}_{\diam\ubox\khence}$};	
	\node[] (logpers) [above left of=logrnfs] {${\sf ITL}^{\sf p}_{\diam\ubox\khence}$};
	
	\path[] 
	(logbasic) edge[] node[pos=0.5,right,font=\scriptsize]{}(loghomeo)
	(logbasic) edge[] node[pos=0.5,left,font=\scriptsize]{}(logrnax)	
	(loghomeo) edge[] node[pos=0.5,right,font=\scriptsize]{}(logrnfs)		
	(logexp) edge[] node[pos=0.5,right,font=\scriptsize]{}(logpers)	
	(logrealax) edge[] node[pos=0.5,right,font=\scriptsize]{}(logrnfs)		
	(logrnax) edge[] node[right,right,pos=0.5,font=\scriptsize]{}(logexp)
	(logrnax) edge[] node[right,pos=0.5,font=\scriptsize]{}(logrealax)		
	(logrnfs) edge[] node[right,pos=0.5,font=\scriptsize]{}(logpers)	
	;								
	\end{tikzpicture}
	\caption{Inclusions between the semantically defined  logics; arrows point from the smaller logic to the larger one.
	See Table \ref{tableClasses} for the definitions of the relevant classes of dynamical systems.}
	\label{fig:sem}
\end{figure}
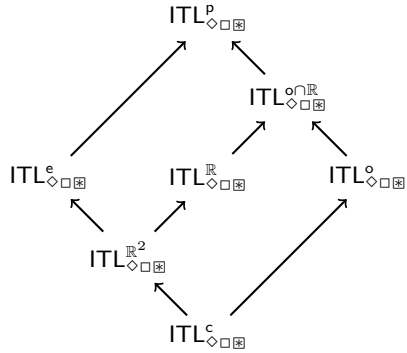

\section{Independence}\label{SecInd}

In this section we will use our soundness results to show that many of the logics we have considered are pairwise distinct.
We begin by showing that the weak logics (based on $\khence$) are in fact weaker than their strong counterparts.
Indeed, certain key $\sf LTL$ principles are not valid for logics with $\khence$.
Recall that the semantics for $\khence$ and $\ubox$ are given in Definition \ref{DefSem}.

\begin{proposition}\label{propWeakSem}
The following are not valid over $\mathbb R$.
\begin{multicols}3
\begin{enumerate}

\item $ {\khence} p \to \tnext {\khence} p$,\\

\item $ {\khence} \tnext p \to \tnext {\khence} p$, and\\

\item ${\khence} p \to {\khence}{\khence} p$.

\end{enumerate}
\end{multicols}
\end{proposition}

\begin{proof}
Let $\mathcal M = (\mathbb R,S,\val\cdot)$, where $S$ is defined as follows:
\item \begin{eqnarray*}
S(x) & = & \begin{cases} 0\textrm{, if }x \leq 0\\
2x\textrm{ if }x > 0 \end{cases}
\end{eqnarray*}
and $\val{p} = (-\infty, 1)$. Note that $S^{-n}(\val{p}) = (-\infty, 1/2^n)$, for $n \geq 0$. Thus,
\[\textstyle\bigcap\limits_{n\geq 0}S^{-n} \val{p} = \textstyle\bigcap\limits_{n\geq 1}S^{-n} \val{p}  = (-\infty, 0].\]
So, $\val{ {\khence}  p}  = \val{{\khence} \tnext p}  =  (-\infty, 0)$.
Hence, $\val{\tnext {\khence}  p} = \val{{\khence} {\khence}  p}  =  \varnothing $, and
\[\val{{\khence}  p \to \tnext {\khence} p} = \val{{\khence} \tnext p \to \tnext{\khence} p} = \val{{\khence} p \to {\khence} {\khence}  p} = (0, \infty) \neq \mathbb R.\]
\end{proof}

Next, we note that the formulas $\rm CD$ and $\rm BI$ separate Kripke semantics from the general topological semantics.

\begin{proposition}\label{PropInd}
The formulas $\AxCons pq$ and $\AxBInd pq$ are not valid over the class of invertible dynamical systems based on~$\mathbb R$, hence $\logrnfs \not \vdash \AxCons pq$ and $\logrnfs \not \vdash \AxBInd pq$.
\end{proposition}

\begin{proof}
Define a model $\mathcal M = (\mathbb R,S,\val\cdot)$ on $\mathbb R$, with $S(x) =  2x$, $\lb p \rb =  (-\infty,1)$ and $\lb q \rb = (0,\infty)$.
Clearly $\val{p\vee q} = \mathbb R$, so that $\val{\ubox(p \vee q)} = \mathbb R$ as well.

Let us see that $ \mathcal M, 0 \not \models \AxCons pq $.
Since $\mathcal M, 0 \models {\ubox(p \vee q)}$, it suffices to show that $\mathcal M, 0 \not \models {\ubox p \vee \diam q}$. It is clear that $\mathcal M, 0 \not \models {\diam q}$ simply because $S^n(0) = 0 \not \in \val q$ for all $n$. Meanwhile, by Example \ref{ExBoxOnR}, $\mathcal M, 0 \models {\ubox p}$ if and only if $\val p = \mathbb R$, which is not the case. We conclude that $
\mathcal M, 0 \not \models \AxCons pq.
$

To see that $
\mathcal M, 0 \not \models \AxBInd pq
$
we proceed similarly, where the only new ingredient is the observation that $\mathcal M, 0  \models \ubox (\tnext q\rightarrow q)$. But this follows easily from the fact that if $\mathcal M,x \models \tnext q$, then $x > 0$ so that $\mathcal M,x \models q$, hence $\val{\tnext q \to q} = \mathbb R$.
\end{proof}

\begin{remark}
Proposition \ref{PropInd} also holds for $t_{\khence} ( \AxCons pq) $ and $t_{\khence} ( \AxBInd pq )$.
However, by Proposition \ref{propBoxHomeo}, these are equivalent to their counterparts with $\ubox$ over the class of invertible systems, so we do not need to mention them separately. A similar comment holds when working over the class of dynamic posets.
\end{remark}

The Fischer Servi axioms are also not valid in general, as already shown in \cite{BalbianiToCL}.
From this and the soundness of $\loghomeo$ (Theorem \ref{ThmSoundFS}), we immediately obtain that they are not derivable in $\logbasic$.

\begin{figure}[H]

\centering

\begin{tikzpicture}[->,auto,minimum size=1mm,node distance=1.5cm]

\node[state,minimum size=4mm,text width=1pt,line width = 1] (n1) {};
\node[state,minimum size=4mm,text width=1pt] (n2) [right of=n1] {};
\node[state,minimum size=4mm,text width=1pt] (n3) [above of=n2] {$p$};
\path[-> ] (n1) edge node{$S$}(n2)
		(n2) edge[loop right] node{$S$} (n2)
		(n3) edge[loop right] node{$S$} (n3)
	  (n2) edge node[right] {$\peq$} (n3);

\end{tikzpicture}

%
%
%
%
%
%
%
%
%
%
%
%
%
%
%
%
%
%
%
%
%
%
%
%
%

\caption{An expanding poset model falsifying both Fischer Servi axioms. Propositional variables that are true on a point are displayed; only one point satisfies $p$ and no point satisfies $q$. It can readily be checked that ${\rm FS}_\mtnext(p,q)$ and ${\rm FS}_\diam(p,q)$ fail on the highlighted point on the left. Note that $S$ is continuous but not open, as can easily be seen by comparing to Figure \ref{FigCO}.}\label{FigIMLA}
\end{figure}
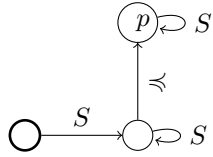

\begin{proposition}\label{prop:nvalid:iltl}\

\begin{enumerate}

\item ${\rm FS}_\diam (p,q)$ is not valid over the class of expanding posets, hence $\logexp \not \vdash {\rm FS}_\diam (p,q)$ and $\logexp \not \vdash {\rm FS}_\mtnext (p,q)$.

\item ${\rm FS}_\diam (p,q)$ and $t_{\khence}( {\rm FS}_\diam (p,q) )$ are not valid over $\mathbb R$, hence $\logrealax \not \vdash {\rm FS}_\diam (p,q)$ and $\logrealax \not \vdash {\rm FS}_\mtnext (p,q)$.

\end{enumerate}
\end{proposition}

\begin{proof}
For the first claim, let us consider the model $\cl M =\left(W,\peq,S,V\right)$ defined by
\begin{enumerate*}[label=\arabic*)]
	\item $W = \lbrace w,v,u\rbrace$;
	\item $S (w) = v$, $S(v) = v$ and $S (u) = u$;
	\item $v \peq u$;
	\item $\val p = \lbrace u \rbrace$, and
	\item $\val q = \varnothing$
\end{enumerate*}	
\noindent (see Figure \ref{FigIMLA}). Clearly, $\cl M, u \not \models p \to q$, so $\cl M, v \not \models p \to q$. By definition, $\cl M, w \not \models \ubox \left(p \to q\right)$; however, $\cl M, w  \models \diam p \to \ubox q$, since the negation of each antecedent holds, so $\cl M, w  \not \models \left( \diam p \rightarrow \ubox q \right) \rightarrow \ubox \left(p \rightarrow q\right)$.

For the second we let $(\mathbb R,S,\val\cdot)$ be a model based on $\mathbb R$ with $S\colon \mathbb R \to \mathbb R$ given by $S(x) = 0$ for all $x$, $\val p = (0,\infty)$, and $\val q = \varnothing$.
Then we have that $\val {\diam p} = (0,\infty)$, so that $- 1 \in \val{ \neg \diam p}$ and hence $ -1 \in \val{  \diam p \to \ubox q}$.
However, if $U$ is an $S$-invariant neighbourhood of $-1$ then $0 = S(-1) \in U$, but $0\not \in \val { p\to q } = (-\infty,0)$, hence $-1\not \in \val{\ubox(p\to q)}$.
It follows that $-1\not \in \val{{\rm FS}_\diam (p,q)}$.
Similar reasoning shows that $-1\not \in \val{ t_{\khence} (  {\rm FS}_\diam (p,q)) }$.
\end{proof}

\begin{remark}
As mentioned previously, \cite{Yuse2006} present a Hilbert-calculus which yields a sub-logic of $\logbasicb$. They also present a Gentzen-style calculus and conjecture that their two calculi prove the same set of formulas.
However, \cite{KojimaNext} show that the formula ${\rm FS}_\mtnext(p,q)$ is derivable in this Gentzen calculus, while Proposition \ref{prop:nvalid:iltl} shows that it is not derivable in $\logbasicb$.
Hence, the two calculi are not equivalent.
\end{remark}

Now we show that our axioms for Euclidean spaces are not valid in general.
In particular, $\rm CEM$ is valid for $\mathbb R$, but it is not valid for higher-dimensional spaces.
In view of Theorem \ref{theoR2vsITLe}, it suffices to show that it is not valid over the class of expanding posets.

\begin{lemma}\label{lemmCemNotRn}
The formula ${\rm CEM} $ is not valid on the class of expanding posets, hence $\logexp \not \vdash {\rm CEM}$.
\end{lemma}

\proof
Consider the model $\mathcal M =(W,{\peq},S,\val\cdot)$, where $W=\{w_0,w_1,v_0,v_1,v_2\}$, $\val p =\{v_2\}$ and $\val q = \{v_1,v_2\}$, and for $x_i,y_j\in W$, $x_i\peq y_j$ if and only if $x=y$ and $i\leq j$, and $S(x_i)=v_i$ (see Figure \ref{figCEM}).
Then, it is not hard to check that $ {\rm CEM}  (p,q) = ( \neg \tnext p \wedge \tnext \neg \neg p ) \rightarrow ( \tnext q \vee \neg \tnext q )$ fails at $w_0$.
\endproof

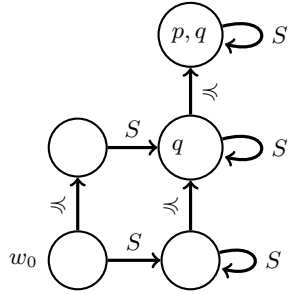
\begin{figure}[H]\centering
		\begin{tikzpicture}[thick,->,auto,font=\small,node distance=1.5cm]
		\node[state,minimum size=3.5mm,text width=5mm, label=west:${w_0}$] (x) {	};
		\node[state,minimum size=3.5mm,text width=5mm] (y)[right of =x] {};
		\node[state,minimum size=3.5mm,text width=5mm] (z)[above of = x] {};
		\node[state,minimum size=3.5mm,text width=5mm] (t)[above of =y] {$q$};
		\node[state,minimum size=3.5mm,text width=5mm] (s)[above of =t] {$p,q$};

		\path[->,very thick] (x) edge[] node{$\peq$}(z)
							 (y) edge[] node{$\peq$}(t)	
							 (x) edge node{$S$}(y)
							 (z) edge node{$S$}(t)
					    	 (t) edge[] node[right] {$\peq$} (s)
					    	 (t) edge[loop right] node[] {$S$} (t)
					    	 (s) edge[loop right] node[] {$S$} (s)
					    	 (y) edge[loop right] node[] {$S$} (y)
					    	 ;	
					    	 	
		\end{tikzpicture}		
		\caption{Model falsifying ${\rm CEM}  (p,q)$ at ${ w_0}$}\label{figCEM}
\end{figure}

Similarly, $\AxConsm p $, which is valid on all Euclidean spaces, is not valid on all dynamical systems, even those based on $\mathbb Q$.

\begin{lemma}
The formula $\AxConsm p $ is not valid on the class of invertible systems based on $\mathbb Q$, hence $\loghomeo\not\vdash \AxConsm p$.
\end{lemma}

\proof
Recall that
$\AxConsm \varphi\psi  = \ubox( p \vee \neg p) \to \diam p \vee \ubox \neg p$.
Let $S$ be given by $S(x)=x+1$.
Define a set
\[D=\mathbb Q\cap \bigcup_{n\in \mathbb N}\left (n-\textstyle\frac 1{n+\pi},n+\textstyle\frac 1{n+\pi}\right)\]
and let $\val p=\mathbb Q\setminus D$. It is readily verified that $\frac 1{n+\pi}\not\in \mathbb Q$ for any $n\in \mathbb N$, and hence
\[\mathbb Q\cap \left (n-\textstyle\frac 1{n+\pi},n+\textstyle\frac 1{n+\pi}\right)=\mathbb Q\cap \left [n-\textstyle\frac 1{n+\pi},n+\textstyle\frac 1{n+\pi}\right],\]
so that $D$ is both closed and open in $\mathbb Q$. It follows that $\val {\neg p}=D$, and hence $\val{p\vee \neg p}=\mathbb Q$; but $\mathbb Q$ is open and $S$-invariant, so $\val{\ubox(p\vee \neg p)}=\mathbb Q$ as well. In particular, $0\in \val{\ubox(p\vee \neg p)}$.

Moreover, we claim that
\begin{enumerate}[label=(\alph*)]
\item\label{ItAExam} $0\not\in \val{\diam p}$, but
\item\label{ItBExam} if $x \in (0,\nicefrac 12)$ and $n>\nicefrac 1x$, then $S^n(x)\in \val{p}$.
\end{enumerate}
Indeed, for \ref{ItAExam} we see that any $n\in\mathbb N$, $S^n(0) = n \in D$, while for \ref{ItBExam}, if $x \in (0,\nicefrac 12)$ and $n>\nicefrac 1x$, then
\[S^n(x) = n+ x \in \left ( n+\frac 1{n+\pi},   n+\frac 12 \right ) 
\subseteq \left ( n+\frac 1{n+\pi},  (n+1) - \frac 1{n+\pi} \right ) 
,\]
so $S^n(x) \not \in D$.
If $U$ is an $S$-closed neighbourhood of $0$, $U$ contains some $x \in (0,\nicefrac 12)$. From \ref{ItBExam} it follows that $S^n(x) \not \in \val{\neg p}$, hence $U \not \subseteq \val {   \neg p}$; since $U$ was arbitrary, $0\not \in \val{\ubox \neg \varphi}$.
\endproof

The above independence results are sufficient to see that the only non-trivial inclusions between our axiomatic systems are given by Proposition \ref{propInclusions}.

\begin{theorem}\label{TheoDistinct}
For each of the following families of axiomatically defined logics (see Table \ref{tableLogics}) or semantically defined logics (see Table \ref{tableClasses}) has pairwise distinct elements, and all subset relations are as indicated in Figures \ref{fig:fig} or \ref{fig:sem}.

\begin{enumerate}

\item\label{claimOne} $\logbasic$, $\loghomeo$, $\logexp$, $\logrealax$, $\logrnax$, $\logrnfs$, and $\logpers$;

\item $\logbasica$, $\loghomeoa$, $\logexpa$, $\logrealaxa$, $\logrnaxa$, $\logrnfsa$, and $\logpersa$; and

\item\label{claimThree} ${\sf ITL}^{\sf c}_{\diam\ubox\khence}$, ${\sf ITL}^{\sf o}_{\diam\ubox\khence}$, ${\sf ITL}^{\sf e}_{\diam\ubox\khence}$, ${\sf ITL}^{\mathbb R}_{\diam\ubox\khence}$, ${\sf ITL}^{\mathbb R^2}_{\diam\ubox\khence}$, ${\sf ITL}_{\diam\ubox\khence}^{{\sf o} \cap \mathbb R}$, and ${\sf ITL}^{\sf p}_{\diam\ubox\khence}$.

\end{enumerate}
The logics in the last item may be replaced by their fragments with only one of $\khence$ or $\ubox$.

\end{theorem}

\begin{proof}
For the first item, each arrow from $\Lambda_1$ to $\Lambda_2$ in Figure \ref{figLogics} is labelled by a formula which we have previously shown to belong to $\Lambda_2\setminus \Lambda_1$.
The same formulas may be used to separate the respective logics in the other two items.
The non-trivial subset relations between the logics have been established in Propositions \ref{propInclusions} and Theorems \ref{theoR2R}, \ref{theoR2vsITLe}, and \ref{theoRcomplete}.
\end{proof}

We may also classify $\diam$-free logics.

\begin{theorem}\label{TheoDistinctBox}
For each of the following families of logics, their elements are pairwise distinct, and all subset relations are as indicated in Figures \ref{fig:box}.

\begin{enumerate}

\item\label{claimOne} ${\sf ITL}_\ubox$, ${\sf ITL}^+_\ubox$, ${\sf CDTL}_\ubox$, ${\sf RTL}_\ubox$, and ${\sf CDTL}^+_\ubox$;

\item ${\sf ITL}_\khence$, ${\sf ITL}^+_\khence$, ${\sf CDTL}_\khence$, ${\sf RTL}_\khence$, and ${\sf CDTL}^+_\khence$, or

\item\label{claimThree} ${\sf ITL}^{\sf c}_{\ubox\khence}$, ${\sf ITL}^{\sf o}_{\ubox\khence}$, ${\sf ITL}^{\sf e}_{\ubox\khence}$, ${\sf ITL}^{\mathbb R}_{\ubox\khence}$, and ${\sf ITL}^{\sf p}_{\ubox\khence}$.

\end{enumerate}
The logics in the last item may be replaced by their fragments with only one of $\khence$ or $\ubox$.

\end{theorem}

\begin{proof}
Similar to Theorem \ref{TheoDistinct}, Figure \ref{figLogics} displays formulas separating these logics, except that instances of ${\rm CD}$ should be replaced by ${\rm BI}$.
\end{proof}

\begin{remark}
Note that logics characterized by ${\rm CD}^-$ are not included in the statement of Theorem \ref{TheoDistinctBox}.
In particular, the formula $\AxBInd{\neg p}p$ is already valid over the class of all dyamical systems.
We do not know if the $\diam$-free logic of dynamical systems based on $\mathbb R^2$ is different from that of all dynamical systems.
\end{remark}

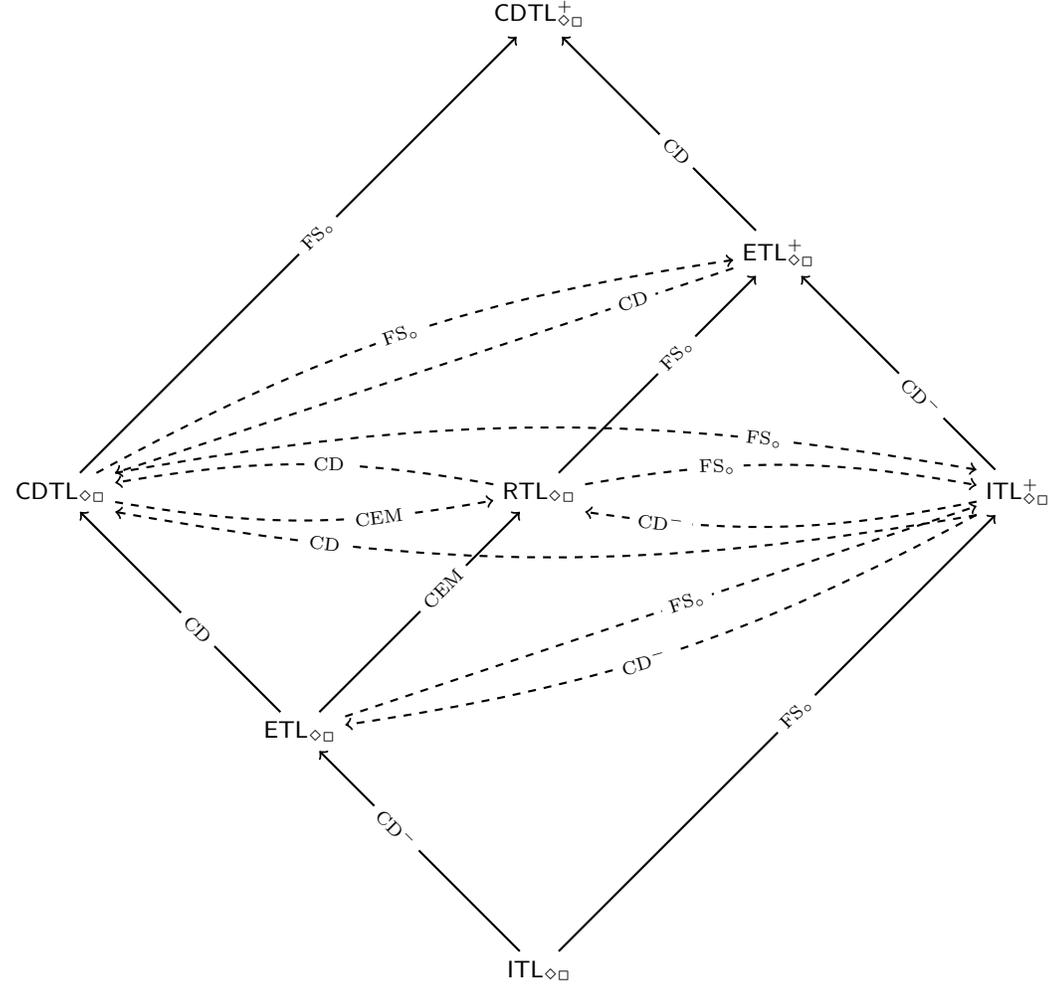
\begin{figure}[h!]\centering
	\begin{tikzpicture}[thick,->,auto,font=\small,node distance=4.5cm]
	\node[] (logbasic) {$\logbasic$};
	\node[] (logrnax) 	[above left of=logbasic] {$\logrnax$};	
	
	\node[] (logrealax) [above right of=logrnax] {$\logrealax$};	
	\node[] (logexp) [above left of=logrnax] {$\logexp$};
	\node[] (logrnfs)  [above right of=logrealax] {$\logrnfs$};		

	\node[] (loghomeo)  [below right of=logrnfs] {$\loghomeo$};

	\node[] (logpers) [above left of=logrnfs] {$\logpers$};
	
	\path[] 
	(logbasic) edge[] node[fill=white,sloped,pos=0.5,right,font=\scriptsize]{${\rm FS}_\mtnext$}(loghomeo)
	(logbasic) edge[] node[fill=white,sloped,pos=0.5,left,font=\scriptsize]{${\rm CD}^-$}(logrnax)	
	(logexp.north east) edge[dashed,bend left=10] node[fill=white,sloped,pos=0.8,left,font=\scriptsize]{${\rm FS}_\mtnext$}(loghomeo.north west)	
	(logrealax) edge[dashed, bend left=10] node[fill=white,sloped,pos=0.4,left,font=\scriptsize]{${\rm FS}_\mtnext $}(loghomeo)								
	(loghomeo) edge[] node[fill=white,sloped,pos=0.5,right,font=\scriptsize]{${\rm CD}^-$}(logrnfs)		
	(loghomeo.south west) edge[dashed,bend left=10] node[fill=white,sloped,pos=0.8,right,font=\scriptsize]{${\rm CD}$}(logexp.south east)
	(logrealax) edge[dashed, bend right=10] node[fill=white,sloped,pos=0.5,right,font=\scriptsize]{${\rm CD}$}(logexp)	
	(logrnfs) edge[dashed] node[fill=white,sloped,pos=0.2,right,font=\scriptsize]{${\rm CD}$}(logexp.north east)
	(logexp) edge[] node[fill=white,sloped,pos=0.5,right,font=\scriptsize]{${\rm FS}_\mtnext$}(logpers)	
	(loghomeo) edge[bend left=10,dashed] node[fill=white,sloped,pos=0.9,right,font=\scriptsize]{${\rm CD}^-$}(logrealax.south east)
	(logexp) edge[bend right=10,dashed] node[fill=white,sloped,left,pos=0.8,font=\scriptsize]{${\rm CEM}$}(logrealax)	
	(logrealax) edge[] node[fill=white,sloped,pos=0.5,right,font=\scriptsize]{${\rm FS}_\mtnext $}(logrnfs)		
	(loghomeo) edge[dashed,bend left=10] node[fill=white,right,sloped,right,pos=0.6,font=\scriptsize]{${\rm CD}^-$}(logrnax)
	(logrnax) edge[dashed] node[fill=white,right,sloped,right,pos=0.5,font=\scriptsize]{${\rm FS}_\mtnext$}(loghomeo)
	(logrnax) edge[] node[fill=white,right,sloped,right,pos=0.5,font=\scriptsize]{${\rm CD}$}(logexp)
	(logrnax) edge[] node[fill=white,sloped,right,pos=0.5,font=\scriptsize]{${\rm CEM}$}(logrealax)		
	(logexp) edge[dashed,bend left=10] node[fill=white,sloped,right,pos=0.45,font=\scriptsize]{${\rm FS}_\mtnext$}(logrnfs)
	(logrnfs) edge[] node[fill=white,sloped,right,pos=0.5,font=\scriptsize]{${\rm CD}$}(logpers)	
	;								
	\end{tikzpicture}
	\caption{Graph displaying the dependences among the different logics studied in this paper. Nodes corresponds to different logic while edges mean two different kinds of relation. Edges of the form  $\Lambda_1 \textrm{ -- }\varphi\rightarrow \Lambda_2$ mean that $\Lambda_1 \subseteq \Lambda_2$ and, moreover, $\varphi \in \Lambda_2 \setminus \Lambda_1$.  Edges of the form $\Lambda_1 \textrm{ - - }\varphi\dashrightarrow \Lambda_2$ mean that $\Lambda_2 \not \subseteq \Lambda_1$ and $\varphi \in \Lambda_2 \setminus \Lambda_1$.}
	\label{figLogics}
\end{figure}

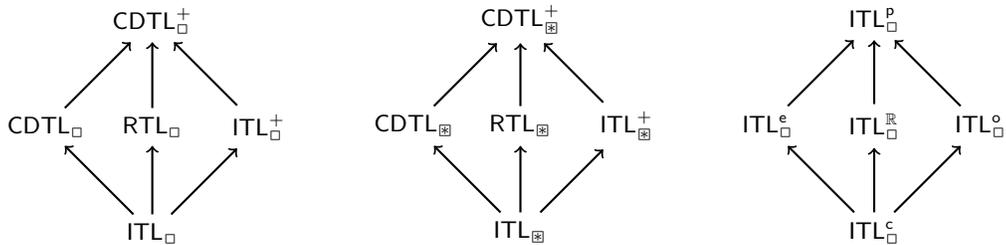
\begin{figure}[h!]\centering
	\begin{tikzpicture}[thick,->,auto,font=\small,node distance=1cm]
	\node[] (logbasic) {${\sf ITL}_\ubox$};
	\node[] (logrnax) 	[above left of=logbasic] {};	
	\node[] (logrealax) [above right of=logrnax] {${\sf RTL}_\ubox$};	
	\node[] (logexp) [above left of=logrnax] {${\sf CDTL}_\ubox$};
	\node[] (logrnfs)  [above right of=logrealax] {};		
	\node[] (loghomeo)  [below right of=logrnfs] {${\sf ITL}^ +_\ubox$};	
	\node[] (logpers) [above left of=logrnfs] {${\sf CDTL}^ +_\ubox$};
	
	\path[] 
	(logbasic) edge[] node[pos=0.5,right,font=\scriptsize]{}(loghomeo)
	(logbasic) edge[] node[pos=0.5,left,font=\scriptsize]{}(logexp)	
		(logbasic) edge[] node[pos=0.5,left,font=\scriptsize]{}(logrealax)	
				(logrealax) edge[] node[pos=0.5,left,font=\scriptsize]{}(logpers)	
	(loghomeo) edge[] node[pos=0.5,right,font=\scriptsize]{}(logpers)		
	(logexp) edge[] node[pos=0.5,right,font=\scriptsize]{}(logpers)	
	;								
	\end{tikzpicture}
	\hfill
		\begin{tikzpicture}[thick,->,auto,font=\small,node distance=1cm]
	\node[] (logbasic) {${\sf ITL}_\khence$};
	\node[] (logrnax) 	[above left of=logbasic] {};	
	\node[] (logrealax) [above right of=logrnax] {${\sf RTL}_\khence$};	
	\node[] (logexp) [above left of=logrnax] {${\sf CDTL}_\khence$};
	\node[] (logrnfs)  [above right of=logrealax] {};		
	\node[] (loghomeo)  [below right of=logrnfs] {${\sf ITL}^ +_\khence$};	
	\node[] (logpers) [above left of=logrnfs] {${\sf CDTL}^ +_\khence$};
	
	\path[] 
	(logbasic) edge[] node[pos=0.5,right,font=\scriptsize]{}(loghomeo)
	(logbasic) edge[] node[pos=0.5,left,font=\scriptsize]{}(logexp)	
		(logbasic) edge[] node[pos=0.5,left,font=\scriptsize]{}(logrealax)	
				(logrealax) edge[] node[pos=0.5,left,font=\scriptsize]{}(logpers)	
	(loghomeo) edge[] node[pos=0.5,right,font=\scriptsize]{}(logpers)		
	(logexp) edge[] node[pos=0.5,right,font=\scriptsize]{}(logpers)	
	;								
	\end{tikzpicture}
	\hfill
		\begin{tikzpicture}[thick,->,auto,font=\small,node distance=1cm]
	\node[] (logbasic) {${\sf ITL}^{\sf c}_\ubox$};
	\node[] (logrnax) 	[above left of=logbasic] {};	
	\node[] (logrealax) [above right of=logrnax] {${\sf ITL}^\mathbb R_\ubox$};	
	\node[] (logexp) [above left of=logrnax] {${\sf ITL}^{\sf e}_\ubox$};
	\node[] (logrnfs)  [above right of=logrealax] {};		
	\node[] (loghomeo)  [below right of=logrnfs] {${\sf ITL}^{\sf o}_\ubox$};	
	\node[] (logpers) [above left of=logrnfs] {${\sf ITL}^{\sf p}_\ubox$};
	
	\path[] 
	(logbasic) edge[] node[pos=0.5,right,font=\scriptsize]{}(loghomeo)
	(logbasic) edge[] node[pos=0.5,left,font=\scriptsize]{}(logexp)	
		(logbasic) edge[] node[pos=0.5,left,font=\scriptsize]{}(logrealax)	
				(logrealax) edge[] node[pos=0.5,left,font=\scriptsize]{}(logpers)	
	(loghomeo) edge[] node[pos=0.5,right,font=\scriptsize]{}(logpers)		
	(logexp) edge[] node[pos=0.5,right,font=\scriptsize]{}(logpers)	
	;								
	\end{tikzpicture}
	\caption{Inclusions between $\diam$-free logics.
}
	\label{fig:box}
\end{figure}

\section{Concluding Remarks and Future Perspectives}\label{SecConc}

We have proposed a natural `basic' intuitionistic temporal logic, $\logbasic$, along with possible extensions including Fischer Servi or constant domain axioms, and weakened versions obtained by modifying the fixed-point axioms for `henceforth'. 
We have seen that relational semantics validate the constant domain axiom, leading us to consider a wider class of models based on topological spaces, with two possible interpretations for `henceforth': the {\em weak henceforth,} $\khence$, and the {\em strong henceforth,} $\ubox$.
With this, we have shown that the logics $\logbasic$, $\logexp$, $\loghomeo$ and $\logpers$ are sound for the class of all dynamical systems, of all dynamical posets, of all open dynamical systems, and of all persistent dynamical posets, respectively, which we have used in order to prove that the logics are pairwise distinct.
We have also shown that the logics $\logrealax$, $\logrnax$, and $\logrnfs$, based on Euclidean spaces, are distinct from any of the above-mentioned logics.
We have performed a similar analysis for logics using $\khence$ instead of $\ubox$.

Of course this immediately raises the question of completeness, which we have not addressed. Specifically, the following are left open.

\begin{question}\label{questOne}
Are the logics:
\begin{enumerate}[label=(\alph*)]

\item $\logbasic$, $\logbasica$, and $\logbasicb$ complete for the class of dynamical systems?

\item $\logexp$ and $\logexpb$ complete for the class of expanding posets?

\item $\loghomeo$, ${\sf ITL}^{\sf o}_\diam$ and $\loghomeob$ complete for the class of open dynamical systems?

\item ${\sf ITL}^\mathbb R_\mtnext$ complete for the class of systems based on $\mathbb R$?

\item $\logrnax$, $\logrnaxa$ complete for the class of systems based on Euclidean spaces?

\item $\logrnfs$ complete for the class of systems based on Euclidean spaces with a homeomorphism?

\item $\logpers$, ${\sf ITL}^{\sf p}_\diam$ and $\logpersb$ complete for the class of persistent posets?

\end{enumerate}

\end{question}

We already know that ${\sf ITL} _\diam$ is sound and complete for the class of expanding posets and for Euclidean spaces \cite{DieguezCompleteness}. However, the completeness of ${\sf ITL}^+_\diam$ and ${\sf ITL}^{\sf p}_\diam$ is likely to be a more difficult problem than that of ${\sf ITL} _\diam$, as in these cases it is not even known if the set of valid formulas is computably enumerable, let alone decidable.

\begin{question}
Are any of the logics $\Lambda$, $\Lambda_\diam$, or $\Lambda_\ubox$ with $\Lambda \in \{\itlp,\itlo\}$ decidable and/or computably enumerable?
\end{question}

A negative answer is possible for any of these logics, since that is the case for their classical counterparts \cite{wolter} and these logics do not have the finite model property \cite{BoudouCSL}.
Nevertheless, the proofs of non-axiomatizability in the classical case do not carry over to the intuitionistic setting in an obvious way, and these remain challenging open problems.

Note that the semantic counterpart for $\logrnax$ used in Theorem \ref{TheoDistinct} is ${\sf ITL}_{\diam\ubox}^{{\sf o} \cap \mathbb R} $.
We could have used ${\sf ITL}^{{\sf o} \cap \{ \mathbb R^n : n\geq 1\}}_{\diam\ubox}$ instead, as $\logrnax$ is also sound for this class.
This raises the following.

\begin{question}
Is every formula falsifiable on some $\mathbb R^n$ with a homeomorphism also falsifiable on $\mathbb R$?
\end{question}

Note that we have not considered weak logics with $\rm CD$ or $\rm FS$.
However, this is only due to the fact that the topological semantics we have considered do not yield semantically-defined logics which satisfy the latter axioms without also satisfying $ \logbasic$.
It may yet be that semantics for such logics may be defined using other classes of dynamical systems.
In particular, our techniques do not show whether the weak and standard logics coincide in these cases.

\begin{question}
Is the logic $\logpers$ distinct from $ \wlogbasic + {\rm CD} + {\rm FS} $?
\end{question}

We conjecture that an affirmative answer could be given using more general algebraic semantics, but we leave this for future work.
Finally, we remark that while we have not considered logics over the full language, it is possible to study logics which combine $\ubox$ and ${\khence}$.
Over dynamic posets or over open dynamical systems such an extension would be uninteresting since both operators are equivalent, but over the class of all dynamical systems, Lemma \ref{lemmStrongtoWeak} suggests defining
\[{\sf ITL}_{\diam\ubox{\khence}} \eqdef \logbasic + \logbasica + \ubox p \to {\khence} p.\]
This leaves us with one final question.

\begin{question}
Is the logic ${\sf ITL}^{\sf c}_{\diam\ubox{\khence}}$ decidable, and does it enjoy a natural axiomatization?
\end{question}

\bibliography{TopoLTL}

\end{document}